\documentclass[11pt]{llncs}

\usepackage{graphicx} \usepackage{times} \usepackage{mathptmx}
\usepackage{subfigure}
\usepackage{epsfig,xspace}

\usepackage{amsmath,amssymb}
\usepackage{enumerate}

\usepackage{times}
\usepackage{wide}
\usepackage{algorithm} \usepackage{algorithmic} %\usepackage{amsthm}
\usepackage{todonotes}

\usepackage{verbatim}
\usepackage{graphicx}
\newcommand{\ie}{\emph{i.e.}, }
\newcommand{\eg}{\emph{e.g.}, }
\newcommand{\dlpt}{\textsc{Dlpt}\xspace}
\newcommand{\copif}{\textsc{CoPIF}\xspace}
\newcommand{\copifs}{\textsc{CoPIFs}\xspace}
\renewenvironment{proof}{{\it Proof. } }{{\hfill $\Box$}\vspace{.5pc}}

 \newtheorem{lem}{Lemma}

\frontmatter          % for the preliminaries

\addtocmark{Optimization in a Self-Stabilizing Service Discovery Framework for Large Scale Systems} % additional mark in the TOC

\title{Optimization in a Self-Stabilizing Service Discovery Framework for Large Scale Systems%
}

\mainmatter              % start of the contributions

\titlerunning{Optimization in Self-Stabilizing Service Discovery}  % abbreviated title (for running head)
\author{
Eddy Caron\inst{1}
\and 
Florent Chuffart\inst{1}
\and 
Anissa Lamani\inst{2}
\and
Franck Petit\inst{3}
}

\institute{
University of Lyon. LIP Laboratory. UMR CNRS - ENS Lyon - INRIA - UCB Lyon 5668, France
\and
MIS Lab. Universit\'e de Picardie Jules Verne, France
\and
 LIP6 CNRS UMR 7606 - INRIA - UPMC Sorbonne Universities, France
}

\authorrunning{Eddy Caron {\it et al}.}   % abbreviated author list (for running head)
\tocauthor{Eddy Caron, Florent Chuffart, Anissa Lamani, Franck Petit}

\begin{document}

\maketitle  

%% C'est bien Eaas et non EAAS 

\begin{abstract} 
  Ability to find and get services is a key requirement %one of the major key 
  in the
  development of large-scale distributed systems.  We consider dynamic
  and unstable environments, namely Peer-to-Peer (P2P) systems.  In
  previous work, we designed a service discovery solution called
  Distributed Lexicographic Placement Table (\dlpt), based on a
  hierarchical overlay structure.  A self-stabilizing version was
  given using the Propagation of Information with Feedback (PIF)
  pa\-ra\-digm. In this paper, we introduce the self-stabilizing \copif
  (for Collaborative PIF) scheme.  An algorithm is provided with its
  correctness proof.  We use this approach to improve a distributed
  P2P framework designed for the services discovery.  Significantly
  efficient experimental results are presented.
\end{abstract}

%\begin{IEEEkeywords}
%  P2P, service discovery, distributed systems, \textsc{Spades}, \textsc{Sbam}, \textsc{Dlpt}.
%\end{IEEEkeywords}

\section{Introduction}

Computing abilities (or {\em services}) offered by large distributed
systems are constantly increasing. Cloud environment grows in this
way.  Ability to find and get these services (without the need for a
centralized server) is a key requirement %one of the major key 
in the development of such
systems.  %The demand to handle flexibility and ability complexity of
Service discovery facilities in distributed systems led to the development of
various overlay structures built over Peer-to-Peer (P2P) systems, \eg
%A REVOIR LE CAS ECHEANT...
\cite{%CAN,Pastry,Chord,mercury,pht:podc04,%% Remove pht:podc04 des 2
  CDT06,MAAN,SchmidtP04,ShuOTZ05}. % to quote only a few.
   Some
of them rely on spanning tree
structures~\cite{CDT06,ShuOTZ05}, mainly to handle range
queries, automatic completion of partial search strings, and to extend
to multi-attribute queries.

Although fault-tolerance is a mandatory feature of systems targeted
for large scale platforms (to avoid data loss and to ensure proper
routing), tree-based distributed structures, including tries, offer only 
a poor robustness in dynamic environment. The crash of one or more nodes 
may lead to the loss of stored objects, and may split the tree into several subtrees. 

The concept of self-stabilization~\cite{SSDolev} is a general technique to
design distributed systems that can handle arbitrary transient faults.
A self-stabilizing system, regardless of the initial state of the
processes and the initial messages in the links, is guaranteed to converge
to the intended behavior in finite time.

In~\cite{CDPT08}, a self-stabilizing message passing
protocol to maintain prefix trees over practical P2P networks is introduced.
%Robustness of this framework is implemented in a dynamic environment 
%and experimentally validated in~\cite{CIT2011}. 
%%%%%%%
% On ne peut pas presenter les choses sur le fait qu'on puisse ou pas
% faire confiance aux données.  De toutes façons, tant qu'on est 
% SELF (ou non pas SNAP), ON NE PEUT PAS FAIRE CONFIANCE AUX DONNEES!
% Il faut simplement dire qu'on améliore les PIF en utilisant la concurrance.
%%%%%%%
% FCh : Mais il n'y avait pas de PIF avant ?!
%%%%%%%
%However, the protocol provided in~\cite{CDPT08} being self-stabilizing, 
%it is very difficult to trust the information provided by the distributed 
%data structure, namely during
%the stabilization phases. As in many applications based on self-stabilizing distributed 
%data structures, checking the truthfulness of provided information 
%is very desirable to increase the quality of service (QoS) of the framework.
%
%
%Tree maintenance in~\cite{CDPT08} is carried out in a self-stabilizing manner by the use of PIF waves. 
%\todo{Attention! pas de PIF dans ~\cite{CDPT08}}
The protocol is based on self-stabili\-zing PIF ({\em Propagation of
  Information with Feedback}) waves that are used to evaluate the tree
maintenance progression.  The scheme of PIF %~\cite{Seg83}
can be informally described as follows: a node, called {\em
  initiator}, initiates a PIF wave by broadcasting a message {\em m}
into the network.  Each non-initiator node acknowledges to the
initiator the receipt of {\em m}. The wave terminates when the root
has received an acknowledgment from all other nodes.  In arbitrary
distributed systems, any node may need to initiate a PIF wave. Thus,
any node can be the initiator of a PIF wave and several PIF protocols
may run concurrently (in that case, every node maintains locally a data
structure per initiator).% based on its identity).

\paragraph{\textbf{Contribution}.} %Our contribution is threefold.  
We first present the scheme of {\em collaborative} PIF (referred as \copif ). %, for short).
The main thrust of this scheme is to ensure that different waves may collaborate to improve the 
overall parallelism of the mechanism of PIF waves. In other words, the waves merge together so 
that they do not have to visit parts of the network already visited by other waves. 
Of course, this scheme is interesting in environments were several PIF waves may run concurrently. 
Next, we provide a self-stabilizing \copif protocol with its correctness proof. 
To the best of our knowledge, it is the first self-stabilizing solution for this problem.  
Based on the snap-stabilizing PIF algorithm in~\cite{DBLP:journals/dc/BuiDPV07},
it merges waves initiated at different points in the network.
In the worst case where only one PIF wave runs at a time, our scheme does not 
slow down the normal progression of the wave.  
Finally, we present experimental results showing the efficiency of our scheme 
use in a large scale P2P tree-based overlay designed for the services discovery.

\paragraph{\textbf{Roadmap}.}
The related works are presented in Section~\ref{sec:RelatedWork}. Section~\ref{sec:P2PServiceDiscoveryFramework} provides the conceptual and computational models of our framework. 
In Section~\ref{sec:Algorithm}, we present and prove the correctness of our self-stabilizing 
collaborative protocol. 
In Section~\ref{sec:Eval},
experiments show the benefit of the \copif approach. 
Finally, 
concluding remarks are given in Section~\ref{sec:Conclusion}.

\section{Related Work}
\label{sec:RelatedWork}

\subsection{Self-stabilizing Propagation of Information}

PIF wave algorithms have been extensively proposed in the area of 
self-sta\-bi\-li\-za\-tion, \eg \cite{AG90,AKMPV93,DBLP:journals/dc/BuiDPV07,DBLP:conf/icdcs/CournierDPV02,V94a} 
to quote only a few. 
Except \cite{AKMPV93,DBLP:conf/icdcs/CournierDPV02,V94a}, all the above solutions 
assume an underlying self-stabilizing rooted spanning 
tree construction algorithm.  
The solutions in \cite{DBLP:journals/dc/BuiDPV07,DBLP:conf/icdcs/CournierDPV02} have the extra desirable 
property of being snap-stabilizing.
A \emph{snap-stabilizing protocol} guarantees that the system always maintains 
the desirable behavior.  
This property is very useful for wave algorithms and other algorithms
that use PIF waves as the underlying protocols.
The basic idea is that, regardless of the initial configuration of the system, 
when an initiator starts a wave, the messages and the 
tasks associated with this wave will work as expected in a normal computation.
A snap-stabilizing PIF is also used in \cite{CDPT10} to propose a snap-stabilizing service 
discovery tool for P2P systems based on prefix tree. 

\subsection{Resource Discovery}

The resource discovery in P2P environments has been intensively
studied~\cite{Meshkova:2008}. Although DHTs~\cite{CAN,Pastry,Chord}
were designed for very large systems, they only provide rigid
mechanisms of search. A great deal of research went into finding ways
to improve the retrieval process over structured peer-to-peer
networks. Peer-to-peer systems use different technologies to support
multi-attribute range
queries~\cite{mercury,MAAN,SchmidtP04,ShuOTZ05}. Trie-structured
approaches outperform others in the sense that logarithmic (or
constant if we assume an upper bound on the depth of the trie) latency
is achieved by parallelizing the resolution of the query in several
branches of the trie.

\subsection{Trie-based related work}
 
Among trie-based approaches, Prefix Hash Tree (PHT)~\cite{pht:podc04}
dynamically builds a trie of the given key-space (full set of possible
identifiers of resources) as an upper layer mapped over any DHT-like
network. Fault-tolerance within PHT is delegated to the DHT
layer. Skip Graphs, introduced in~\cite{AspnesS2003}, are similar to
tries, and rely on skip lists, using their own probabilistic
fault-tolerance guarantees.  P-Grid is a similar binary trie whose
nodes of different sub-parts of the trie are linked by shortcuts like
in Kademlia~\cite{MM02}. The fault-tolerance approach used in
P-Grid~\cite{pgrid} is based on probabilistic replication.

In our approach, the \dlpt was initially designed for the purpose of
service discovery over dynamic computational grids and aimed at
solving some drawbacks of similar previous approaches. An advantage of
this technology is its ability to take into account the heterogeneity
of the underlying physical network to build a more efficient tree
overlay, as detailed in~\cite{CDT08}.

\section{P2P Service Discovery Framework}

\label{sec:P2PServiceDiscoveryFramework}

In this section we present the conceptual model of our P2P service discovery framework and the \dlpt data structure on which it is based.
Next, we convert our framework into the computational model on which our proof is based.

\subsection{Conceptual Model} 
\label{sec:ConcMod}

\emph{The two abstraction layers} that compose our P2P service discovery framework are organized as follow: 
($i$) a $P2P$
network which consists of a set of asynchronous peer (physical
machines) with distinct identifiers. The peer communicate by
exchanging messages. Any peer $P1$ is able to communicate with
another peer $P2$ only if $P1$ knows the identifier of $P2$. The
system is seen as an undirected graph $G=(V,E)$ where $V$ is the set
of peers and $E$ is the set of bidirectional communication
link; ($ii$) an overlay that is built on the P2P system, which is
considered as an undirected connected labeled tree $G'=(V',E')$ where
$V'$ is the set of nodes and $E'$ is the set of links between
nodes. Two nodes $p$ and $q$ are said to be neighbors if and only if
there is a link $(p,q)$ between the two nodes. To simplify the
presentation we refer to the link $(p,q)$ by the label $q$ in the code
of $p$. The overlay can be seen as an indexing system whose nodes are mapped onto the peers of the network. 
%%TEMP REMOVED%%For instance, Figure~\ref{fig:P2PFrameworkDesignExample} shows $3$ peers (Peer 1, Peer 2, Peer 3) that supports $5$-nodes tree ($.$, $a$, $b$, $ba$, $bb$).
Henceforth, to avoid any confusion, the word {\em node} refers to a node of the tree overlay, {\it i.e.}, a logical entity, whereas the word \emph{peer} refers to a physical node part of the $P2P$ system.

%%TEMP REMOVED%%\begin{figure}[b]
%%TEMP REMOVED%%  \centering
%%TEMP REMOVED%%  \mbox{\includegraphics[trim = 0mm 60mm 0mm 60mm,clip, width=.9\linewidth]{figs/P2PFrameworkDesignExample}}
%%TEMP REMOVED%%  \caption{\label{fig:P2PFrameworkDesignExample} P2P Framework Design.}
%%TEMP REMOVED%%\end{figure}

\emph{Reading and writing features} of  our service discovery framework are ensured as follow. Nodes are indexed with service name and resource locations are stored on nodes. So, client requests are treated by any node, rooted to the targeted service labeled node along the overlay abstraction layer, indexed resource locations are returned to the clients or updated . A more detailed description of the implementation of our framework is given in~\cite{CIT2011} and briefly reminded in Section~\ref{sec:Sbam}.

The \emph{Distributed Lexicographic Placement Table} (\dlpt~\cite{CDT06,CDT08}) 
is the hierarchical data structure that ensures request routing across overlay layer. 
\dlpt belongs to the category of overlays that are distributed prefix trees, {\em e.g.}, ~\cite{AspnesS07,RamabhadranRHS04,AbererCDDHPS03}. 
Such overlays have the desirable property of efficiently supporting range queries by parallelizing the searches in 
branches of the tree and exhibit good complexity properties due to the
limited depth of the tree.
More particularly, \dlpt is based on the particular 
\emph{Proper Greatest Common Prefix Tree} (\textsc{Pgcp} tree) overlay structure.
\emph{A Proper Greatest Common Prefix} Tree ({\em a.k.a} radix tree in~\cite{Morrison68}) is a labeled rooted tree such that the following properties are true for every node of the tree:
($i$) the node label is a proper prefix of any label in its subtree;
($ii$) the greatest common prefix of any pair of labels of children of a given node are the same and equal to the node label.

Designed to evolve in very dynamic systems, the \dlpt integrates
a self-stabilization mechanisms~\cite{CDPT08}, providing the
ability to recover a functioning state after arbitrary transient failures.
As such, the truthfulness of information returned to the client needs to be guaranteed. 
We use the PIF mechanism to check whether \dlpt is currently in a recovering phase or not.

%The next section presents our framework from an other point of view: a computational model which our proof will based on. 

\subsection{Computational Model}
\label{ref:CompMod}

In a first step, we abstract the communication model to ease the reading and the explanation of our solution.  
We assume that every pair of neighboring nodes communicate in the overlay by direct reading of variables. % instead of exchange of messages between peers. 
So, the program of every node consists in a set of shared variables (henceforth referred to as variables) and a finite number of actions. 
Each node can write in its own variables and read its own variables and those of its neighbors. Each action is constituted as follow:
%
%$
$<Label>::<Guard> \rightarrow <Statement>$. %$
The guard of an action is a Boolean expression involving the variables of $p$ and its neighbors. The statement is an action which updates one or more variables of the node $p$. Note that an action can be executed only if its guard is true. Each execution is decomposed into steps. 
Let $y$ be an execution and $A$ an action of $p$ ($p$ $\in$ $V$). $A$ is {\em enabled} for $p$ in $y$ if and only if the guard of $A$ is satisfied by $p$ in $y$. Node $p$ is enabled in $y$ if and only if at least one action is enabled at $p$
in $y$.

The state of a node is defined by the value of its variables. The state of a system is the product of the states of all nodes. The local state refers to the state of a node and the global state to the state of the system. Each step of the execution consists of two sequential phases atomically executed:
($i$) Every node evaluates its guard;
($ii$) One or more enabled nodes execute their enabled actions. 
When the two phases are done, the next step begins.

Formal description (Section~\ref{sec:Formal}) and proof of correctness %(Section~\ref{sec:Proof}) 
of the proposed collaborative propagate information feedback algorithm will be done using this computational model.
Nevertheless, experiments are implemented using the classical message-passing model over an actual peer-to-peer
system~\cite{BCC+06,DrostCaCPE10}.

\section{Collaborative Propagation of Information with Feedback Algorithm}
\label{sec:Algorithm}

In this section, we first present an overview of the proposed Collaborative Propagation of Information with Feedback Algorithm (\copif).
Next, we provide its formal description.
%Finally, we prove the correctness of our solution.

\subsection{Overview of the \copif}

\begin{figure}[t]
  \centering
  \subfigure[\label{fig:PIF_Broadcast}Broadcast step.]{\mbox{\includegraphics[trim = 0mm 0mm 0mm 0mm, clip, width=.49\linewidth]{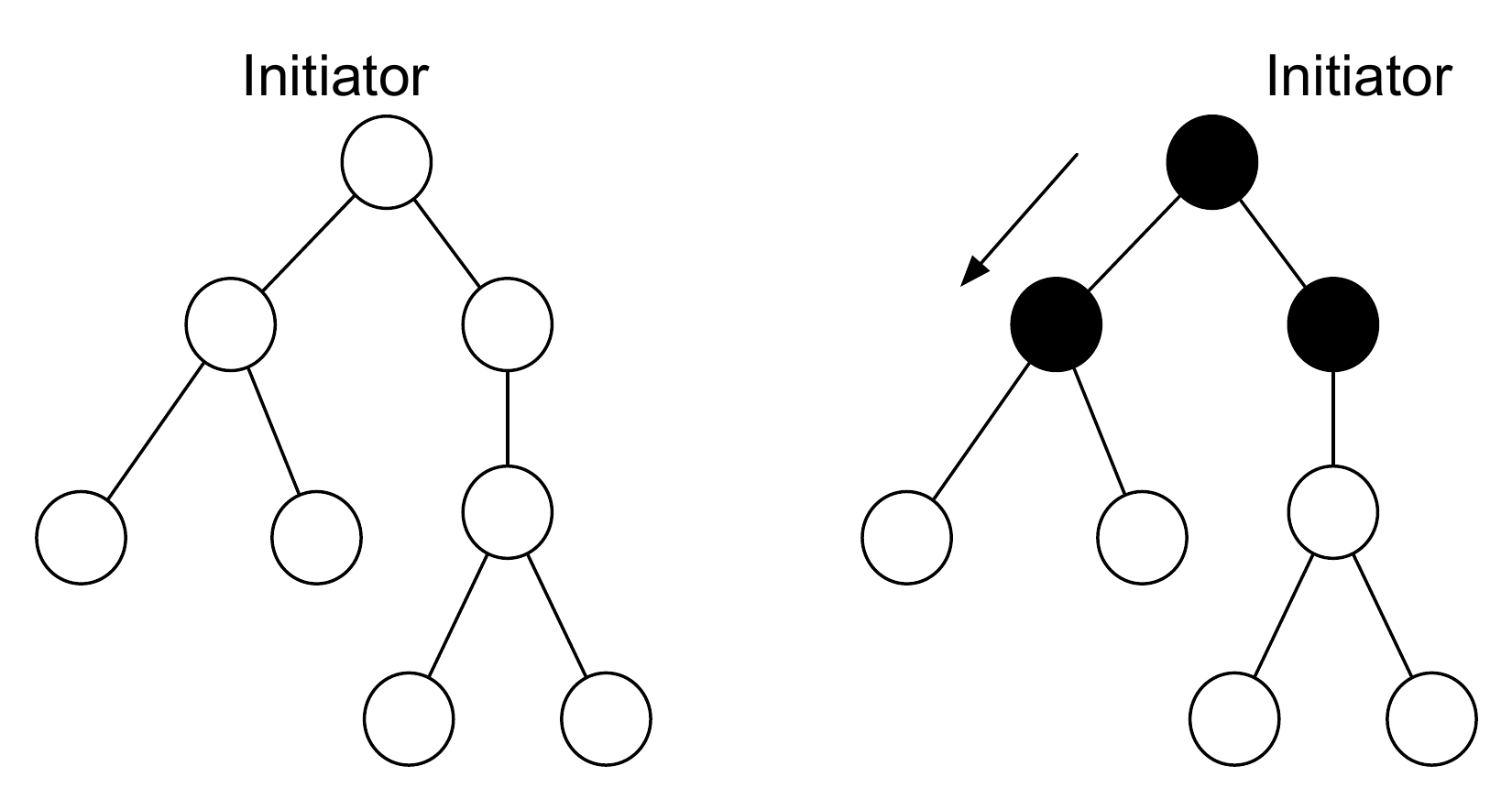}}}
  \subfigure[\label{fig:PIF_Feedback}Feedback step.]{\mbox{\includegraphics[trim = 0mm 0mm 0mm 0mm, clip, width=.49\linewidth]{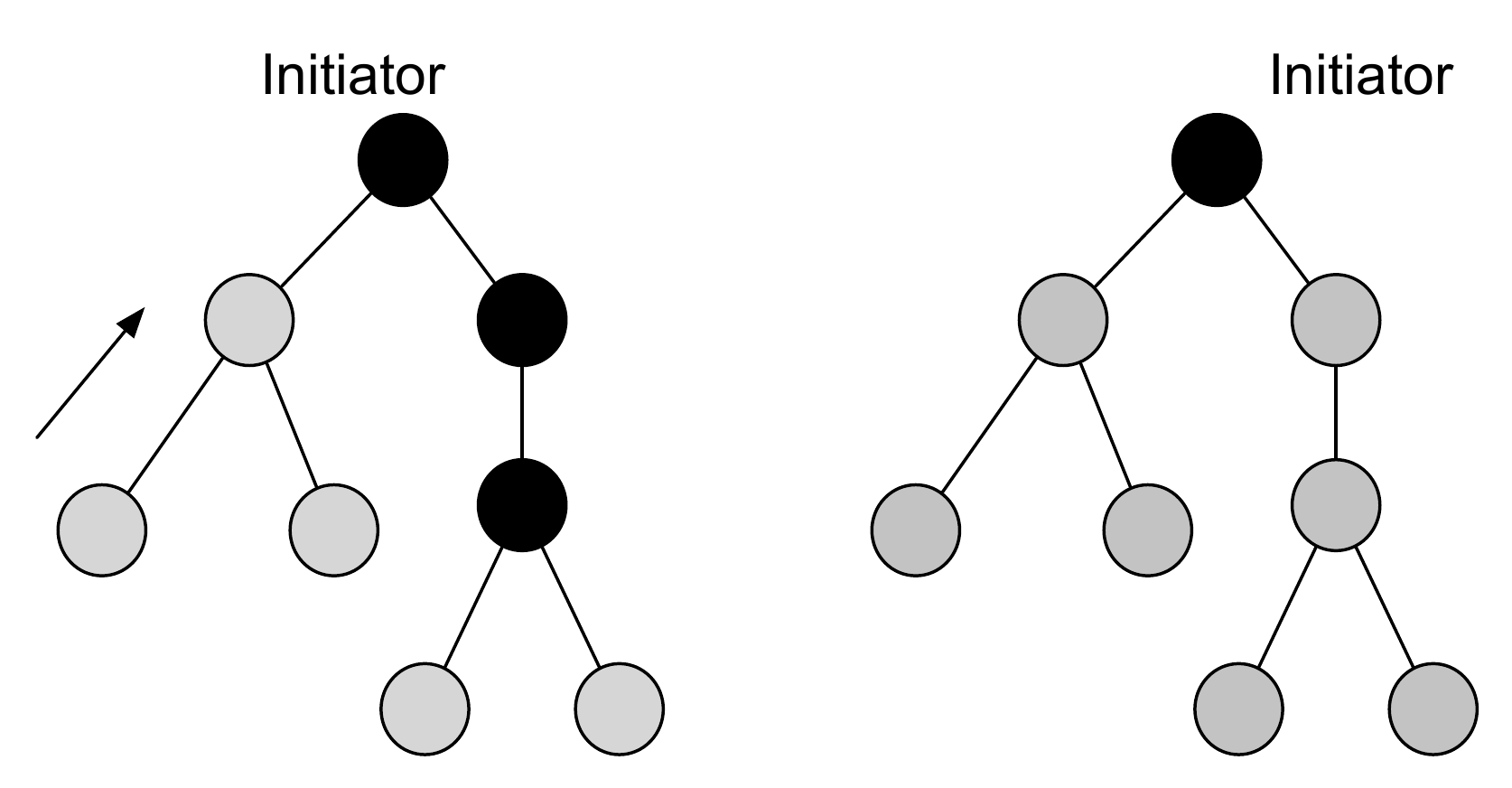}}}
  \caption{\label{fig:PIF}The PIF Wave.}
\end{figure}

Before explaining the idea behind \copif, let us first recall the well-known PIF wave execution. 
Starting from a configuration where no message has been broadcast yet, a node, also called \emph{initiator}, initiates
the broadcast phase and all its descendant except the leaf participate in this task by sending also the broadcast
message to their descendants. % (Figure~\ref{fig:PIF_Broadcast}). 
Once the broadcast message reaches a leaf node of the network, they notify their ancestors of the end of the broadcast
phase by initiating the feedback phase. % (Figure~\ref{fig:PIF_Feedback}). 
During both broadcast and feedback steps, it is possible to collect information or perform actions on the entire data structure. 
Once all the nodes of the structure have been reached and returned the feedback message, the initiator retrieve
collected information and executes a special action related to the termination of the PIF-wave. In the sequel, we will refer to this mechanism as {\it classic}-PIF.

%%%SSS%%%\begin{figure}
%%%SSS%%%  \centering
%%%SSS%%%  \subfigure[\label{fig:PIF_Broadcast}Broadcast step.]{\mbox{\includegraphics[trim = 0mm 0mm 0mm 0mm, clip, width=.35\linewidth]{figs/PIF-Broadcast}}}
%%%SSS%%%  \subfigure[\label{fig:PIF_Feedback}Feedback step.]{\mbox{\includegraphics[trim = 0mm 0mm 0mm 0mm, clip, width=.35\linewidth]{figs/PIF-Feedback}}}
%%%SSS%%%  \caption{\label{fig:PIF}The PIF Wave.}
%%%SSS%%%\end{figure}

%\todo[inline]{Anissa: Peux tu enlever les lgendes "broadcast" et "feedback" sur Figure~\ref{fig:PIF}. Et mettre "Initiator" en plus gros.}

However, the PIF mechanism is a costly broadcast mechanism that involves the whole platform. 
In this paper, we aim to increase the parallelism of the PIF by making several PIF waves collaborating together. 
Let us now define some notions that will be used in the description of our solution:

Let an ordered alphabet $A$ be a finite set of letters. 
Lets define $\prec$ an order on $A$. 
A non empty word $w$ over $A$ is a finite sequence of letters $a_1$, ... , $a_i$, ..., $a_l$ such as $l>0$. 
The concatenation of two words $u$ and $v$, denoted as $uv$, is equal to the word $a_1$, ..., $a_i$,  ..., $a_k$, $b_1$, ..., $b_j$, ..., $b_l$ such that
$u=a_1$, ..., $a_i$,  ..., $a_k$ and $v=b_1$, ..., $b_j$, ..., $b_l$. 
%Let $\xi$ be the empty word such that for every word $w$, $w\xi = w$. 
	%The length of a word $w$, denoted by $\mid w\mid$, is equal to the number of letters of $w$.
	%$\mid \xi \mid=0$. 
	A word $u$ is a prefix (respectively, proper prefix) of a word $v$ if there exists a word $w$ such that $v = uw$  (respectively, $v=uw$ and $u\ne v$). 
The {\em Greatest Common Prefix} (respectively, {\em Proper Greatest Common Prefix}) of $w_1$ and $w_2$, denoted $GCP(w_1, w_2)$ (respectively $PGCP(w_1,w_2)$, is the longest prefix $u$ shared by $w_1$ and $w_2$ (respectively, such that $\forall i \geq 1$, $u\ne w_i$).

Let us now describe the outline of the proposed solution through the P2P framework use case.

\paragraph{\textit{\textbf{Use Case}}.} 
The idea of the algorithm is the following: When a user is looking for a service, it sends a request to the \dlpt to check whether the service exists or not. 
Once the request is on one node of the \dlpt, it is routed according to the labelled tree in the following manner: let $l_{request}$ be the label of the service requested by the user and let $l_p$ be the label of the current node $u_{p}$. %hosting the user's request. 
In the case of $PGCP(l_{request},l_p)$ is true,  $u_{p}$ checks whether there exists a child $u_{q}$ in the \dlpt having a label $l_q$ such that $PGCP(l_{request},l_q)$ 
is satisfied. 
If such a node exists, then $u_{p}$ forwards the request to its child $u_{q}$. 
Otherwise ($PGCP(l_{request},l_p)$ is not satisfied), if we keep exploring the sub-tree routed in $u_p$, the service will not be found. 
$u_p$ sends in this case the request to its father node in the \dlpt. 
By doing so, either 
($i$) the request is sent to one node $u_p$ such that $l_p=l_{request}$, or 
($ii$) the request reaches a node $u_p$ such that it cannot be routed anymore.
In the former case, the service being found, a message containing the information about the service is sent to the
user.  In the latter case, the service has not been found and the message ``no information about the service'' is sent to the user.
However, the node has no clue to trust the received information or not. 
%
%\paragraph{Problems} 
In other words, in the former case,  
$u_p$ does not know whether it contains the entire service information or if 
a part of the information is on a node being at a wrong position in the tree due to transient faults. 
In the latter case, $u_p$ does not know whether the service is really not supported by the system or if the service
is missed because it is at a wrong position. 

%\paragraph{Solution} 
In order to solve this problem, $u_p$ initiates a PIF wave to check the state of all the nodes part of the \dlpt. 
Note that several PIF waves can be initiated concurrently since many requests can be made in different parts of
the system.  The idea of the solution is to make the different PIF waves collaborating in order to check whether 
the tree is under construction or not. 
For instance, assume that two PIF waves, $PIF1$ and $PIF2$, are running concurrently on two different parts of the
tree, namely on the subtrees $T1$ and $T2$, respectively. 
Our idea is to merge $PIF1$ and $PIF2$ so that $PIF1$ (respectively, $PIF2$) do not traverse $T2$ (resp., $T1$) by
using data collected by $PIF2$ (resp.$PIF1$).  Furthermore, our solution is required to be self-stabilizing. 

\paragraph{\textit{\textbf{\copif.}}} Basically, the \copif scheme is a mechanism enabling the collaboration between different PIF waves. 
Each node $u_p$ of the \dlpt has a state variable $S_p$ that includes three parameters
$S_p=(Phase,id_{f},id_{PIF})$.  Parameter $id_{PIF}$ refers to the identifier of the PIF wave which consists of the couple ($id_{peer}$,$l_{u_i}$), where $id_{peer}$ is the identifier of the peer hosting the node $u_i$ that initiated the PIF wave and $l_{u_i}$ is the label of the node $u_i$. 
The value $id_{f}$ refers to the identifier of the neighbor from which $u_p$ received the broadcast. 
It is set at NULL in the case $u_p$ is the initiator.  
$Phase$ can have four values: $C$, $B$, $FC$ and $FI$. The value $C$ (\emph{Clean}) denotes the initial state of any node before it participates in a PIF wave. 
The value $B$ (\emph{Broadcast}) or $FC$ (\emph{Feedback correct}) or $FI$ (\emph{Feedback incorrect}) means that the node is part of a PIF wave. 
Observe that in the case there is just a single PIF wave that is executed on the \dlpt, then its execution is similar to the previously introduce {\it classic}-PIF. 
%In this case we are sure that the initiator of the PIF wave will get the right answer (Recall that, in this case the PIF is snap stabilizing). 

When more than one PIF wave are executed, four cases are possible while the progression of the
\copif wave. %---refer to Figure~\ref{fig:CoPIF4case}.  
First~($i$), if there is a node $u_p$ in the $C$ state having only
one neighboring node $q$ in the $B$ state and no other neighboring node in the $FI$ or $FC$ state, then $p$ changes
its state to $B$. Second~($ii$), if there exists a leaf node $u_p$ in the $C$ state having a neighbor $u_q$ in the $B$
state, then $u_p$ changes its state to $FC$ (resp. $FI$) if its position in the \dlpt is correct (resp. incorrect). 
Next~($iii$), if there is a node $u_p$ in $C$-phase having two neighboring nodes $u_q$ and $u_{q'}$ in the $B$ state with
different $id_{PIF}$ then, $u_p$ changes its state to $B$ and sets its $id_{f}$ to $u_q$ such that the $id_{PIF}$ of
$u_q$ is smaller than $id_{PIF}$ of $u_{q'}$.  Finally~($iv$), if there exists a node $u_p$ that is already in the $B$ 
state such that its $id_{f}$ is $u_q$ and
there exists another neighboring node $u_{q'}$ which is in the $B$ state with a smaller $id_{PIF}$ and a different
$id_{f}$, then $u_p$ changes its father by setting $id_{f}$ at $u_{q'}$. 

Notice that in the fourth cases, $u_q$ (previously, the $id_{f}$ of $u_p$) will have to change its
$id_{f}$ as well since it has now a neighbor $u_p$ in the $B$ state with a smaller $id_{peer}$. 
By doing so, the node $u_i$ that initiated a PIF wave with a smaller $id$ will change its $id_{peer}$. 
Similarly, notice that $u_i$ is not an initiator anymore. Hence it changes its $id_{f}$ from NULL to the 
$id$ of its neighbor with a smallest $id_{PIF}$. So, only one node will get the answer (the feedback of the \copif),
this node being the one with the smallest $id_{PIF}$.  Therefore, when an initiator sets its $id_{f}$ to a value
different from NULL (as $u_i$ previously), it sends a message to the new considered initiator (can be deduced from
$id_{PIF}$) to subscribe to the answer.  So, when an initiator node receives the feedback that indicates the state of the
tree, it notifies all its subscribers of the answer.

\begin{figure}
  \centering
  \mbox{\includegraphics[trim = 40mm 64mm 40mm 75mm, clip, width=.95\linewidth]{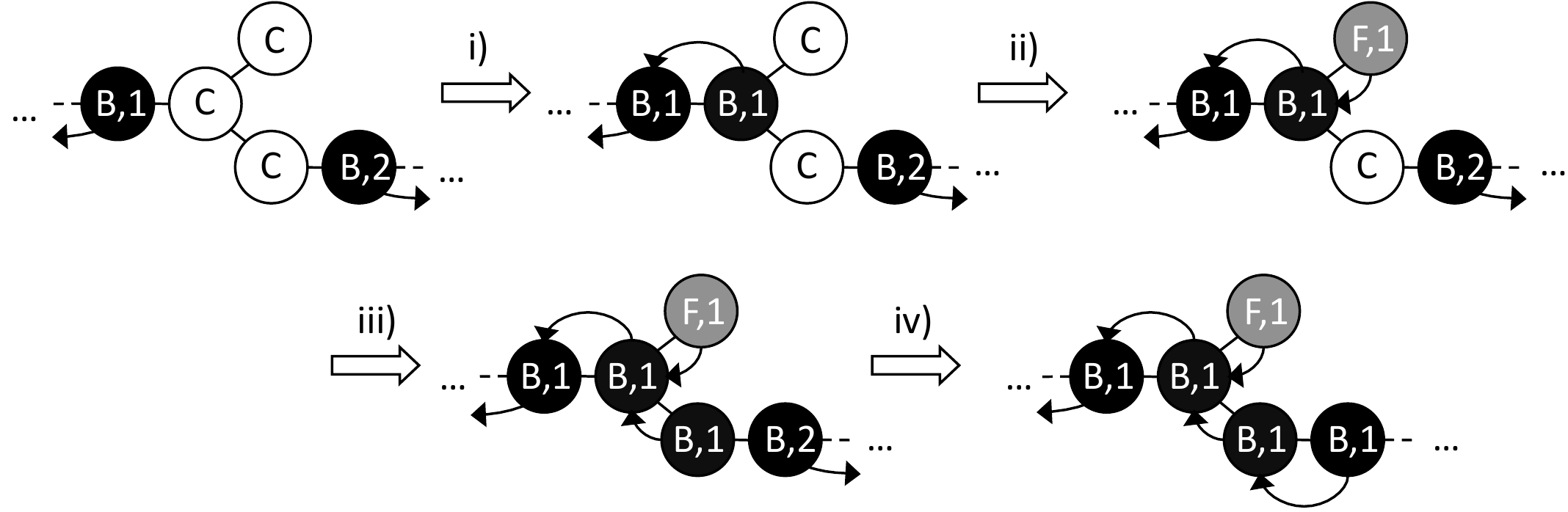}}
  \caption{\label{fig:CoPIF4case} 4 CoPIF wave transition.}
\end{figure}

%%%SSS%%%\begin{figure}
%%%SSS%%%  \centering
%%%SSS%%%  \mbox{\includegraphics[trim = 40mm 64mm 40mm 75mm, clip, width=.65\linewidth]{figs/CoPIF4case}}
%%%SSS%%%  \caption{\label{fig:CoPIF4case} Four \copif wave transitions.}
%%%SSS%%%\end{figure}

	%there will be a node $u_p$ that will have too neighbors $u_q$ and $u_{q'}$ that will be in the broadcast
	%state with different $id_{PIF}$. In this case, $u_p$ will change its state to $B$ and set its $id_{\mbox{father}$ to $q$ such as the $id_{PIF}$ of $u_q$ is smaller than $id_{PIF}$ of $u_{q'}$. Note that now, the node $u_{q'}$ is in a broadcast mode and it is part of a $PIF2$ having one neighbor ($u_{p}$) in the broadcast state as well but with a smaller $id_{PIF}$. $u_{q'}$ will change it father by choosing $u_p$ as its new father and it sends a message to the initiator of $PIF2$ to notify it that the answer will be available at the node $id_PIF1$. When the initiator of $PIF2$ receives such a message it will send a message to the initiator of $PIF1$ to let it know that it is interested by the answer.   
	
\subsection{Formal Description}
\label{sec:Formal}

In the following we first define the data and variables that are used for the description of our algorithm. We then present the formal description in Algorithm~\ref{algo:PIF}.

\begin{itemize}
   \item {Predicates}
             \begin{itemize}
               \item $Request_{PIF}$: Set at true when the peer wants to initiate a PIF wave (There is a $Service-Request$ which could not find the desired service).
               %\item $NoPIF_{peer}$: Set at true when the peer has at least one node in the DLPT that initiated a PIF wave   
            \end{itemize}
   \item {Variables}     
              \begin{itemize}
                \item $S_p=(A,q,q')$: refers to the state of the node $p$ such that: A corresponds to the phase of the
PIF wave $p$ is in. $A \in \{B, FI, FC, C\}$ for respectively Broadcast, Feedback State-Incorrect, Feedback
State-Correct, Clean. q refers to the identity of the peer that initiates the PIF wave. q' refers to the identity of
the neighboring node of $p$ in the \dlpt from which $p$ got the Broadcast.
                \item $N_p$: refers to the set of the identities of the nodes that are neighbor to $p$
                \item $State_{DLPT}$: refers to the state of the \dlpt
              %  \item $State_p$: refers to the local state of the node $p$
                 \item $min_p$: $q$ $\in N_p$, $S_q=(B,id_q,z)$ $\wedge$ $z\ne p$ $\wedge$ $id_p=min$\{$id_{q'}$, $q'\in N_p$, $S_{q'}=(B,id_{q'},z')$ $\wedge$ $z'\ne p$\}.
             \end{itemize}  
  \item{Functions}
                 \begin{itemize}
                   \item Send(@dest,@source, Msg): @source sends the message {\em Msg} to @dest.
                   \item Add(Mylist, item): Add to my list the subject {\em item}.
                 \end{itemize}
\end{itemize}

 Character {\tt '-'} in the algorithm means {\em any value}.

\begin{algorithm}[htbp]
\caption{\copif \label{algo:PIF}}
\begin{scriptsize}

   \begin{itemize}
   
\item{\textbf{PIF Initiation}}   
    \begin{itemize}
       \item $\textbf{R1}$: $Request_{PIF}$ $\wedge$ $S_p=(C,-,-)$ $\wedge$ $\forall$ $q\in N_p$, $S_q\ne (-,-,p)$ $\rightarrow$ $S_p=(B,(id_{peer},l_p),NULL)$, $State_{DLPT}=Unknown$
       \item  $\textbf{R2}$: $Request_{PIF}$ $\wedge$ $S_p=(B,id,q)$ $\wedge$ $q\ne NULL$ $\rightarrow$ $Send(@id, idpeer, Interested)$, $State_{DLPT}=Unknown$\\
    \end{itemize}
   
%$PIF-Request_p=true$ $\wedge$ $S_p=(C,?,?,?)$ $\rightarrow$ $S_p=(B,p,idle,idle)$
\item{\textbf{Broadcast propagation}}
\begin{itemize}
\item$\textbf{R3}$: $S_p=(C,-,-)$ $\wedge$ $\neg Request_{PIF}$ $\wedge$ $\exists$ $q \in N_p$, ($S_q=(B,k,-)$ $\wedge$ $q=min_p$ $\wedge$ $\neg \exists$ $q'\in N_p$, $q\ne q'$, $S_{q'}=(B,k,-)$) $\wedge$ $\forall$ $q''\in N_p$, $S_{q''}\ne$ $(FI \vee FC,-,p)$ $\rightarrow$ $S_p=(B,k,q)$\\
\end{itemize}

\item\textbf{Father-Switch}
\begin{itemize}
\item$\textbf{R4}$: $\exists$ $q \in N_p$, $S_p=(B,id,q)$ $\wedge$ $\exists$ $q' \in N_p$, ($q\ne q'$ $\wedge$ $S_{q'}=(B,id',?)$ $\wedge$ $id'<id$ $\wedge$ $q=min_p$ ) $\rightarrow$ $S_p=(B,id',q')$\\
\end{itemize}

\item\textbf{initiator resignation}
\begin{itemize}
\item $\textbf{R5}$: $S_p=(B,id,NULL)$ $\wedge$ $\exists$ $q \in N_p$, $S_{q}=(B,id',?)$ $\wedge$ $id'<id$ $\wedge$ $q=min_p$) $\rightarrow$ $S_p=(B,id',q)$, $Send(@id, id'peer, Interested)$
\end{itemize}
\item\textbf{Feedback initiation}
    \begin{itemize}
    \item $\textbf{R6}$: $|N_p|=1$ $\wedge$ $State=Correct$ $\wedge$ $\exists$ $q\in N_p$, $S_q=(B,id,?)$ \textbf{$\rightarrow$} $S_p=(FC,id,q)$
    \item $\textbf{R7}$: $|N_p|=1$ $\wedge$ $State=Incorrect$ $\wedge$ $\exists$ $q\in N_p$, $S_q=(B,id,?)$ \textbf{$\rightarrow$} $S_p=(FI,id,q)$\\
    \end{itemize}
    
\item\textbf{Feedback propagation}
\begin{itemize}
\item $\textbf{R8}$: $\exists$ $q\in N_p$, $S_q=(B,id,-)$ $\wedge$ $S_p=(B,-,q)$ $\wedge$ $\forall$ $q'\in N_p/\{q\}$, $S_q'=(FC,-,p)$ \textbf{$\rightarrow$} $S_p=(FC,id,q)$ 
\item $\textbf{R9}$: $\exists$ $q\in N_p$, $S_q=(B,id,-)$ $\wedge$ $S_p=(B,-,q)$ $\wedge$ $\forall$ $q'\in N_p/\{q\}$, $S_q'=(FI \vee FC,-,p)$ $\wedge$ $\exists$ $q''\in N_p/\{q\}$  $S_q'=(FI,?,p)$ \textbf{$\rightarrow$} $S_p=(FI,id,q)$ \\
\end{itemize}

\item\textbf{Cleaning phase initiation}
\begin{itemize}
\item $\textbf{R10}$: $\forall$ $q\in N_p$, $S_q=(FC,id',p)$ $\wedge$ $S_p=(B,id,NULL)$ $\wedge$ $id=(id_{peer},l_p)$ \textbf{$\rightarrow$} $State_{DLPT}=Correct$, $Request_{PIF}=false$, $S_p=(C,NULL,NULL)$, $Send(@ListToContact$, 'DLPT Correct'), $State_{DLPT}=Unknown$
\item $\textbf{R11}$: $\forall$ $q\in N_p$, $S_q=(FI \vee FC,id',p)$ $\wedge$ $S_p=(B,id,NULL)$ $\wedge$ $id=(id_{peer},l_p)$ $\wedge$ $\exists$ $q''\in N_p/\{q\}$  $S_q'=(FI,-,p)$ \textbf{$\rightarrow$} $State_{DLPT}=Incorrect$, $Request_{PIF}=false$, $S_p=(C,NULL,NULL)$, $Send(@ListToContact$, 'DLPT Incorrect'), $State_{DLPT}=Unknown$\\
\end{itemize}

\item\textbf{Cleaning phase propagation}
\begin{itemize}
\item $\textbf{R12}$: $\exists$ $q \in N_p$, $S_p=(FI \vee FC, id,q)$ $\wedge$ ($S_q=(C,-,-)$ $\vee$ $q=NULL$) \textbf{$\rightarrow$} $S_p=(C, NULL, NULL)$\\
\end{itemize}

\item \textbf{Correction Rules}
\begin{itemize}
\item $\textbf{R13}$: $S_p=(B,id,NULL)$ $\wedge$ $id\ne (idpeer,l_p)$ \textbf{$\rightarrow$} $S_p=(C,NULL,NULL)$ %$id'=(id_{peer}, l_p)$, $S_p=(B,id',NULL)$, $Add(ListToContact,id)$
\item $\textbf{R14}$: $\exists q \in N_p$ $S_p=(FI \vee FC,id,q)$ $\wedge$ $\exists$ $q' \in N_p$, $q\ne q'$ $\wedge$ $S_{q'}\ne(FI \vee FC,-,-)$ \textbf{$\rightarrow$} $S_p=(C,NULL, NULL)$
\item $\textbf{R15}$: $S_p=(B,id,q)$ $\wedge$ ($S_q\ne(B,-,-)$ $\vee$ [$S_q\ne(B,id',-)$ $\wedge$ $id'>id$]  \textbf{$\rightarrow$}  $S_p=(C,NULL,NULL)$  %$id'=(id_{peer}, l_p)$, $S_p=(B,id',NULL)$, $Add(ListToContact,id)$
\item $\textbf{R16}$: $S_p=(B,-,q)$ $\wedge$ $S_q=(FI \vee FC,-,-)$ \textbf{$\rightarrow$} $S_p=(C,NULL,NULL)$
\item $\textbf{R17}$: $\exists q\in N_p$, $S_p=(B,id,q)$ $\wedge$ $S_q=(B,id,p)$ \textbf{$\rightarrow$} $S_p(C,NULL,NULL)$
\item $\textbf{R18}$: $\exists$ $q,q'\in N_p$, $S_p=(B,id,q)$ $\wedge$ $q\ne q'$ $\wedge$ $S_{q'}=(B,id,z)$ $\wedge$ $z\ne p$ \textbf{$\rightarrow$} $S_p=(F,id,q)$
\item $\textbf{R19}$: $\exists$ $q,q'\in N_p$, $S_p=(C,NULL,NULL)$ $\wedge$ $S_{q}=(B,id,z)$ $\wedge$ $z\ne p$ $\wedge$ $S_{q'}=(B,id,z')$ $\wedge$ $z'\ne p$ \textbf{$\rightarrow$} $S_p=(F,id,q)$
\item $\textbf{R20}$: $\exists$ $q,q'\in N_p$, $S_p=(B,id,q)$ $\wedge$ $S_{q}=(B,id',z)$ $\wedge$ $z\ne p$ $\wedge$ $id'<id$ $\wedge$  $S_{q'}=(B,id'',z')$ $\wedge$ $z'\ne p$ $\wedge$ $id'<id''$ \textbf{$\rightarrow$} $S_p=(B,id',q)$\\
\end{itemize}

\item \textbf{Event: Message reception}
         \begin{itemize}
           \item Message 'idpeer,Interested': $Add(ListToContact,idpeer)$
           \item Message 'Contact id for an answer': $Send(@id, idpeer, Interested)$
           \item Message 'DLPT Correct': $State_{DLPT}=Correct$
           \item Message 'DLPT Incorrect': $State_{DLPT}=Incorrect$
        \end{itemize}
\end{itemize}        
   \end{scriptsize}
\end{algorithm}

\subsection{Correctness Proof.}

%We first show that starting from any arbitrary configuration, the system eventually contains no abnormal sequence,
%\ie incorrect process states due to the unpredictable initial configurations and transient errors.
%Next, we show that each node is able to generate a $PIF$ wave in finite time.  
%Furthermore, all the nodes of the system are visited by the \copif wave. Thus, all of them acknowledge the receipt of the question
%(whether the tree overlay is in a correct state or not) and give an answer to the latter. Finally, one node $p$ of the 
%system receives the answer ($S_p=(B,id,NULL)$). Hence, the following statement holds:

%\begin{theorem}
%\label{th:PIF}
%Algorithm \ref{algo:PIF} is a self-Stabilizing \copif algorithm.
%\end{theorem}

In the following, we prove the correctness of our algorithm. 

Let us first define our self-stabilizing $CoPIF$ wave:
\begin{definition}($CoPIF$ wave)\\
A finite computation $e$ is called $CoPIF$ wave if and only if the following conditions hold:
\begin{itemize}
\item Each node in the system is able to initiate a $PIF$ wave in a finite time.
\item All the nodes of the system are visited by $CoPIF$ wave.
\item Exactly one node of the system (that initiated a $PIF$ wave) receives the acknowledgement from all the other nodes. 
\end{itemize}
\end{definition}
%
%A protocol $P$ satisfies $SP$ if and only if the following requirements are satisfied in every execution of $P$:
%\begin{enumerate}
%\item every peer is able to either initiate a PIF wave or if it can't it will receive the feedback needed information in a finite time.
%\item The $PIF$ wave terminates at only one of the initiators who has the smallest $id_{PIF}$ (let refer to such a node by $Final-PIF$).
%\item $Final-PIF$ receives the right answer (the state of the overlay) in a finite time.
%\end{enumerate}
%\end{specification}

Let us now define some notions that will be used later.

\begin{definition} (abnormal sequences of type $A$)\\
We say that a configuration contains an abnormal sequence of type $A$ if there exists a node $p$ of state $S_p=(B,id,q)$ such as one of these conditions holds:

\begin{enumerate}
\item $q=NULL$ $\wedge$ $id\ne (id_{peer},lp)$.
\item $q\ne NULL$ $\wedge$ $S_q=(B,id',z)$ $\wedge$ $id'>id$ $\wedge$ $z\ne p$.
\item $q\ne NULL$ $\wedge$ $S_q=(B,id,p)$.
\item $q\ne NULL$ $\wedge$ $S_q\ne (B,-,-)$.
\end{enumerate}

In the following we refer to each case by type $Ai$ where $1\leq i \leq 4$.
%There
%We say that a configuration contains an abnormal sequence of type $A$ if the following condition holds:
%\begin{center}
%$\exists q,q'\in N_p$, $S_p=(B,id,q)$ $\wedge$ $S_q=(B,id',-)$ $\wedge$ $id< id'$ $\wedge$ $S_q'=(B,id,z)$ $\wedge$ $z\ne p$
%\end{center}
%
%In the same manner, we say that a configuration contains an abnormal sequence of type $B$ if the following condition holds:
%\begin{center}
%$\exists q \in N_p$, $S_p=(B,id,q)$ $\wedge$ $S_q=(B,id,p)$
%\end{center}

\end{definition}

\begin{definition} (abnormal sequences of type $B$)\\
We say that a configuration contains an abnormal sequence of type $B$ if there exists a node $p$ of state $S_p=(C,NULL,NULL)$ such as $\exists$ $q,q' \in N_p$, $S_q=(B,id,z)$ $\wedge$ $S_{q'}=(B,id,z')$ $\wedge$ $z\ne p$ $\wedge$ $z'\ne p$.
\end{definition}

\begin{definition}(Dynamic abnormal sequence)\\
We say that a configuration contains a dynamic abnormal sequence if there exists two nodes $p$ and $q$ such as $S_p=(B,id,z)$ with $z \ne q$ and $S_p=(B,id,z')$ with $z' \ne p$.
\end{definition}

\begin{definition}(Trap sequence)\\
We say that a configuration contains a trap sequence if there exists a sequence of nodes $p_0$, $p_1$, ..., $p_k$ such that the following three conditions hold: $(i)$ $S_{p_{0}}=(B,id,z)$ $\wedge$ $z\ne p_1$. $(ii)$ $S_{p_{k}}=(B,id,z')$ $\wedge$ $z\ne p_{k-1}$. $(iii)$ $\forall$ $1\leq i \leq k-1$, $S_{p_{i}}=(C,NULL, NULL)$. 
\end{definition}

\begin{definition} (Path)\\
The sequence of nodes $P_{id}=p_0$, $p_1$, ..., $p_k$ is called a path  if $\forall$ $1\leq i \leq k$, $S_{p_{i}}=(B,id,p_{i-1})$ $\wedge$ $S_{p_{0}}=(B,id,NULL)$. $p_0$ is said to be the extremity of the path.
\end{definition}

\begin{definition} (FullPath)\\
For any node $p$ such as $S_p\ne (C,NULL,NULL)$, a unique path $PN= p_0$, $p_1$, ..., $p_k$ is called FullPath, if and only if  $\forall$ $1\leq i \leq k$, $S_{p_{i}}=(B \vee F, -,p_{i-1})$ $\wedge$ $S_{p_{0}}=(B\vee F,-,p)$. $p$ is said to be the extremity of the path.
\end{definition}

\begin{definition} (SubTree)\\
For any node $p$, we define a set of $SubTree(p)$ of nodes as follow: for any node $q$, $q$ $\in SubTree(p)$ if and only if $p$ and $q$ are part of the same FullPath such as $p$ is the extremity of the path.
\end{definition}

In the following, we say that a node $p$ clean its state if it updates its state to $S_p=(C,NULL,NULL)$.
%\begin{definition} (abnormal feedback)\\
%We say that a configuration contains an abnormal feedback if the following condition holds:
%$\exists q \in N_p$, $S_p=(,-,q)$ $\wedge$ $S_q=(FI \vee FC, id,p)$
%\end{definition}

%In the following, we first prove that if the configuration contains an abnormal sequence, then a configuration without any abnormal sequence is reached in finite time.

Let us first show that a configuration without abnormal sequences of type $A$ and $B$ is reached in a finite time.

\begin{lem}\label{AbnSeq}
No abnormal sequence of type $A$ can be created dynamically.
\end{lem}

\begin{proof}
Let consider each abnormal sequence of type $A$ separately: 

\begin{enumerate}
\item %$S_p=(B,id,q)$ $\wedge$ $q=NULL$ $\wedge$ $id\ne (id_{peer},lp)$. 
 Type $A1$. Note that the only rule in the algorithm that allows $p$ to set its state to $S_p=(B,id,NULL)$ is $R1$. Note also that when $R1$ is executed on $p$, $id=(id_{peer},lp)$. Thus we are sure that the abnormal sequence $A1$ is never created dynamically. 
%Note that in this case the state of $p$ indicates that $p$ is the initiator of a $PIF$ wave since $q=NULL$. However, the identity of the $PIF$ wave is different from $(id_{peer},l_p)$. Note from the algorithm that when the initiator initiates a $PIF$ waves $id=(id_{peer},l_p)$. Note also that when an initiator $p$ resigns of being initiator ($p$ has a neighboring node in the broadcast phase having a smaller $id_PIF$), it changes its state as follow: $S_p=(B,id',q')$ such as $id'\ne (id_{peer},lp)$ and $q'\ne NULL$.
\item %$q\ne NULL$ $\wedge$ $S_q=(B,id',z)$ $\wedge$ $id'>id$ $\wedge$ $z\ne p$. 
Type $A2$. Note that the only case where $p$ sets its state to $S_p=(B,id,q)$ such as $q\ne NULL$ is when $S_q=(B,id',z)$ such as $id'<id$ and $z\ne p$ (see Rules $R3$, $R4$ and $R5$). Thus we are sure that no abnormal sequence of type $A2$ is created dynamically. 
\item %$q\ne NULL$ $\wedge$ $S_q=(B,id,p)$. 
Type $A3$. To create this abnormal sequence dynamically either $p$ or $q$ or both are in the clean state $(C,NULL,NULL)$. In the three cases the node that has a clean state $(C,NULL,NULL)$ (let this node be $p$) never change its state to set it to the broadcast phase $(B,id,q)$ when $S_{q}=(C,NULL,NULL)$ $\vee$ $S_{q}=(B,id,p)$. Thus we are sure that no abnormal sequence of type $A3$ is created dynamically.  
\item  %$q\ne NULL$ $\wedge$ $S_q\ne (B,id,-)$. XXX
Type $A4$. The properties of the $PIF$ algorithm ensure that when a node $p$ is in a broadcast phase $(B,id,q)$ (refer
to \cite{DBLP:journals/dc/BuiDPV07}), then $S_q=(B,-,-)$. Thus we are sure that no abnormal sequence of type $A4$ is created dynamically.  
\end{enumerate}
\end{proof}

From Lemmas \ref{AbnSeq}, we can deduce that the number of abnormal sequences does not increase.

\begin{lem}
Every execution of Algorithm \ref{algo:PIF} contains a suffix of configurations
containing no abnormal sequence of type $A1$.
\end{lem}

\begin{proof}
Note that in a case of an abnormal sequence of type $A1$, there exists a node $p$ whom state $Sp=(B,id,q)$ such as $q=NULL$ $\wedge$ $id\ne (id_{peer},lp)$. Note that $R13$ is enabled on $p$. When the rule is executed on $p$, $S_p=(C,NULL,NULL)$ and the Lemma holds. 
\end{proof}

\begin{lem}
Every execution of Algorithm \ref{algo:PIF} contains a suffix of configurations
containing no abnormal sequence of type $A2$.% is reached in a finite time. % $O(1)$ round.
\end{lem}

\begin{proof}
Note that in a case of an abnormal sequence of type $A2$, there exists a node $p$ whom state $Sp=(B,id,q)$ such as $q\ne NULL$ $\wedge$ $S_q=(B,id',z)$ $\wedge$ $id'>id$ $\wedge$ $z\ne p$. Note that $R15$ is enabled on $p$. When the rule is executed $S_p=(C,NULL,NULL)$ and the Lemma holds. 
\end{proof}

\begin{lem}
Every execution of Algorithm \ref{algo:PIF} contains a suffix of configurations
containing no abnormal sequence of type $A3$. % $O(1)$ round.
\end{lem}

\begin{proof}
Note that in a case of an abnormal sequence of type $A3$, there exists a node $p$ whom state $Sp=(B,id,q)$ such as $q\ne NULL$ $\wedge$ $S_q=(B,id,p)$. Note that $R17$ is enabled on $p$. When the rule is executed $S_p=(C,NULL,NULL)$ and the Lemma holds. 
\end{proof}

\begin{lem}
Every execution of Algorithm \ref{algo:PIF} contains a suffix of configurations
containing no abnormal sequence of type $A4$.
%Starting from a configuration that contains an abnormal sequence of type $A4$, a configuration without any abnormal sequence of type $A4$ is reached in a finite time. %$O(h)$ rounds.
\end{lem}

\begin{proof}
 Let $S=n_1,n_2,...,n_k$ be the sequence of node on the tree overlay such as $S_{n_1}=(B,-,p)$ and $S_{n_i}=(B,-,n_{i-1})$. Note that in a case of an abnormal sequence of type $A3$, there exists a node $p$ whom state $Sp=(B,id,q)$ such as $q\ne NULL$ $\wedge$ $S_q\ne (B,-,-)$. In this case $R15$ is enabled on $p$. When the rule is executed $S_p=(C,NULL,NULL)$. Note that the state of $n_1$ is the same as $p$ is the previous round, thus when $n_1$ executes $R15$, $S_{n_1}=(C,NULL,NULL)$. $R15$ becomes then enabled on $n_{2}$ and so on. %Since the longest sequence $S$ is when $k=2h-1$, we are sure that after $O(h)$ rounds a configuration without any abnormal sequence of type $A4$.
 Thus we are sure to reach in a finite time a configuration without any abnormal sequence of type $A4$.
\end{proof}

\begin{lem} \label{noAbsB}
Every execution of Algorithm \ref{algo:PIF} contains a suffix of configurations
containing no abnormal sequence of type $B$.
\end{lem}

\begin{proof}
Note that in the case of an abnormal sequence of type $B$ there is in the configuration at least one node $p$ such that the following properties hold: $S_p=(C,NULL,NULL)$ $\wedge$ $\exists$ $q,q' \in N_p$, $S_q=(B,id,z)$ $\wedge$ $S_{q'}=(B,id,z')$ $\wedge$ $z\ne p$ $\wedge$ $z'\ne p$. Note that in this case $R19$ is enabled on $p$. When the rule is executed, $S_p=(F,id,q)$ (the scheduler will make the choice between $q$ and $q'$). When all the nodes on which $R19$ is enabled execute $R19$. A configuration without 	any abnormal sequence of type $B$ is reached and the lemma holds.  
\end{proof}

Let us show now that a configuration without dynamic abnormal sequences is reached in a finite time. 
\begin{lem} \label{TrapConf}
If the configuration contains a trap sequence, then this sequence was already in the system in the starting configuration.
\end{lem}

\begin{proof}
The proof is by contradiction. We suppose that the trap sequence can be created during the execution using the $PIF$ waves that were initiated after the starting configuration.\\
First of all note that since the identifier of each $PIF$ wave is unique, there is only one initiator for each $PIF$ wave. Thus, there is at most one $PIF$ wave that can be executed for each $id$ (considering only the $PIF$ waves that were executed after the initial configuration). On the other hand, when a $PIF$ wave is executed, the only rules that makes a node change its father $id$ are $R4$ and $R5$. Note that when one of these rules is enabled on $p$, $p$ has a neighboring node $q$ in the broadcast phase with a smaller $PIF$ $id$. When one of these rules is executed, $p$ changes both the identifier of its father and the identifier of the $PIF$ wave. Thus we are sure that if the configuration contains a trap sequence then, this sequence was already in the system in the starting configuration.
\end{proof}

\begin{lem}\label{DynAbn}
No dynamic abnormal sequence is created dynamically eventually. 
\end{lem}

\begin{proof}
The two cases bellow are possible:
\begin{enumerate}
\item \label{Onenode} There is a node $p$ that has a clean state $S_p=(C,NULL,NULL)$ such as it has two neighboring nodes $q$ and $q'$ with $S_q=(B,id,z)$, $S_{q'}=(B,id,z')$, $z\ne p$ and $z'\ne p$. In this case $R19$ is enabled on $p$. When the rule is executed, $p$ sets its state directly to the feedback phase $S_p=(F,id,q \vee q')$ (the adversary will choose between $q$ and $q'$). Thus we are sure that no dynamic abnormal sequence is reached in this case.
\item \label{Twonodes} There are two neighboring nodes $p$ and $p'$
  such as the two following condition hold: $(i)$ $\exists q\in N_p$,
  $S_q=(B,id,z)$ $\wedge$ $z\ne p$ and $(ii)$ $\exists q'\in N_{p'}$,
  $S_{q'}=(B,id,z')$ $\wedge$ $z'\ne p'$. Note that on both $p$ and
  $p'$, $R3$ is enabled. When the rule is executed only on one node,
  we retrieve Case \ref{Onenode}. Thus, no dynamic abnormal
  configuration is reached. In the case $R3$ is executed on both $p$
  and $p'$ then a dynamic abnormal sequence is created. However,
  observe that the latter case cannot happen infinitely often since in a
  correct execution, it is impossible to reach a configuration where
  two $PIF$ waves with the same id are executed on two disjoint sub-trees
  (refer to Lemma \ref{TrapConf}). So we are sure that a limited number of dynamic abnormal sequences can be created (due to the arbitrary starting configuration). Thus we are sure that, after a finite time, no dynamic abnormal configuration is reached. 
\end{enumerate}
From the cases above, we can deduce that no dynamic abnormal sequence is created dynamically eventually and the lemma holds.
\end{proof}

\begin{lem}
Every execution of Algorithm \ref{algo:PIF} contains a suffix of configurations
containing no dynamic abnormal
sequences.
%Starting from a configuration that contains a dynamic abnormal sequence, a configuration without any dynamic abnormal sequence is reached in a finite time. %$O(1)$ round.
\end{lem}

\begin{proof}
From Lemma \ref{DynAbn}, a limited number of dynamic abnormal sequences can be created. Let's consider the system when all the dynamic abnormal sequences have been created (no other Dynamic abnormal sequence can be created). Note that when the configuration contains a dynamic abnormal sequence, there are at least two nodes $p$ and $q$ that are neighbors such as $S_p=(B,id,z)$, $S_q=(B,id,z')$, $z\ne q$ and $z'\ne p$. Note that $R18$ is enabled on both $p$ and $q$. When the rule is executed on at least one of the two nodes (let this node be the node $p$), $S_p=(F,id,z)$ and the lemma holds.  
\end{proof}

From the lemmas above we can deduce that a configuration without any abnormal sequences of (type $A$ and $B$), any dynamic abnormal sequences is reached in a finite time. 

% It is clear that if there is exactly only one $PIF$ wave that is executed on the overlay then, The PIF wave will have the same behavior as \cite{} and we are sure that in this case the Specification \ref{spec:SP} is verified. /*Ne pas oublier de rajouter la specif*/

%%%%%%%%%%%%%%%%%%%%%%%%%%%%%%%%%%%%%%%%%%%%%%%%%%%%%%%%%%%%%%%%%%%%%%%%%%%%%%%%%%%%%%%%%%%%%%%%%%%
In the following we consider the system when there are no abnormal sequences (Dynamic, Type $A$ and $B$). We have the following lemma: 

\begin{lem}\label{PathChange}
Let $P_{id}= p_0$, $p_1$, ..., $p_k$ be a path. If there exist $1 \leq i \leq k$ such as $p_i$ changes its state to $S_{p_{i}}=(B,id',q)$ with $q \ne p_{i-1}$ and $id'<id$ then, $\forall$ $0 \leq j \leq i$, $p_j$ will also updates their state to $S_{p_{j}}=(B,id', p_{j+1})$ in a finite time. % in $O(h)$ rounds.
\end{lem}

\begin{proof}
Note that when $p_i$ changes its state to $S_{p_{i}}=(B,id',q)$ with $q \ne p_{i-1}$ and $id'<id$, $p_{i-1}$ becomes neighbor of a node ($p_i$) that is in the broadcast phase with a smallest $id$, $p_{i-1}$ changes its state to $S_{p_{i-1}}=(B,id', p_i)$ by executing $R4$. Note that now $R15$ now is enabled on $p_{i-1}$ and so on, thus we are sure that $\forall$ $1 \leq j \leq i$, $p_j$ will updates their state to $S_{p_{j}}=(B,id', p_{j+1})$. Note that when $p_1$ updates its state, $R5$ becomes enabled on $p_0$, when $R5$ is executed, $p_0=(B,id,p_1)$. Thus the lemma holds. %On another hand, since the longest path in this case is of length $2h-1$, $O(h)$ rounds are needed and the lemma holds.  
\end{proof}

Let us refer by $Initial-PIF$ waves, the set of $PIF$ waves that were already in execution in the initial configuration (they were not initiated after the faults). Recall that the starting configuration can be any arbitrary configuration. Thus, some $PIF$ waves can be separated by nodes that are in the feedback phase: $PIF_{id}$ and $PIF_{id'}$ are said to be separated by nodes in the feedback phase if there exists a sequence of nodes $p_0, p_1,..., p_k$ such that $S_{p_{0}}=(B,id,z)$ with $z\ne p_1$ and $S_{p_{k}}=(B,id',z')$ with $z'\ne p_{k-1}$ and ($S_{p_{1}}=(F,id,p_0)$ $\vee$ $S_{p_{-1}}=(F,id',p_k)$). Let refer by $PIF$ friendly set $S_i$, the set of nodes that are part of $PIF$ waves that are not separated by nodes in the feedback phase. Note that $1\leq i \leq k$.

In the following, we say that there is a Partial-Final sequence in a $PIF$ friendly set $S_i$, if there exists a node $p$ such that $S_p=(B,id,NULL)$ and $\forall$ $q \in S_i/\{p\}$, $q \in SubTree(p)$ and $S_q=(F,-,-)$ (refer to Figure \ref{PartialFig}).

\begin{figure}
   \centering
   \includegraphics[width=5cm]{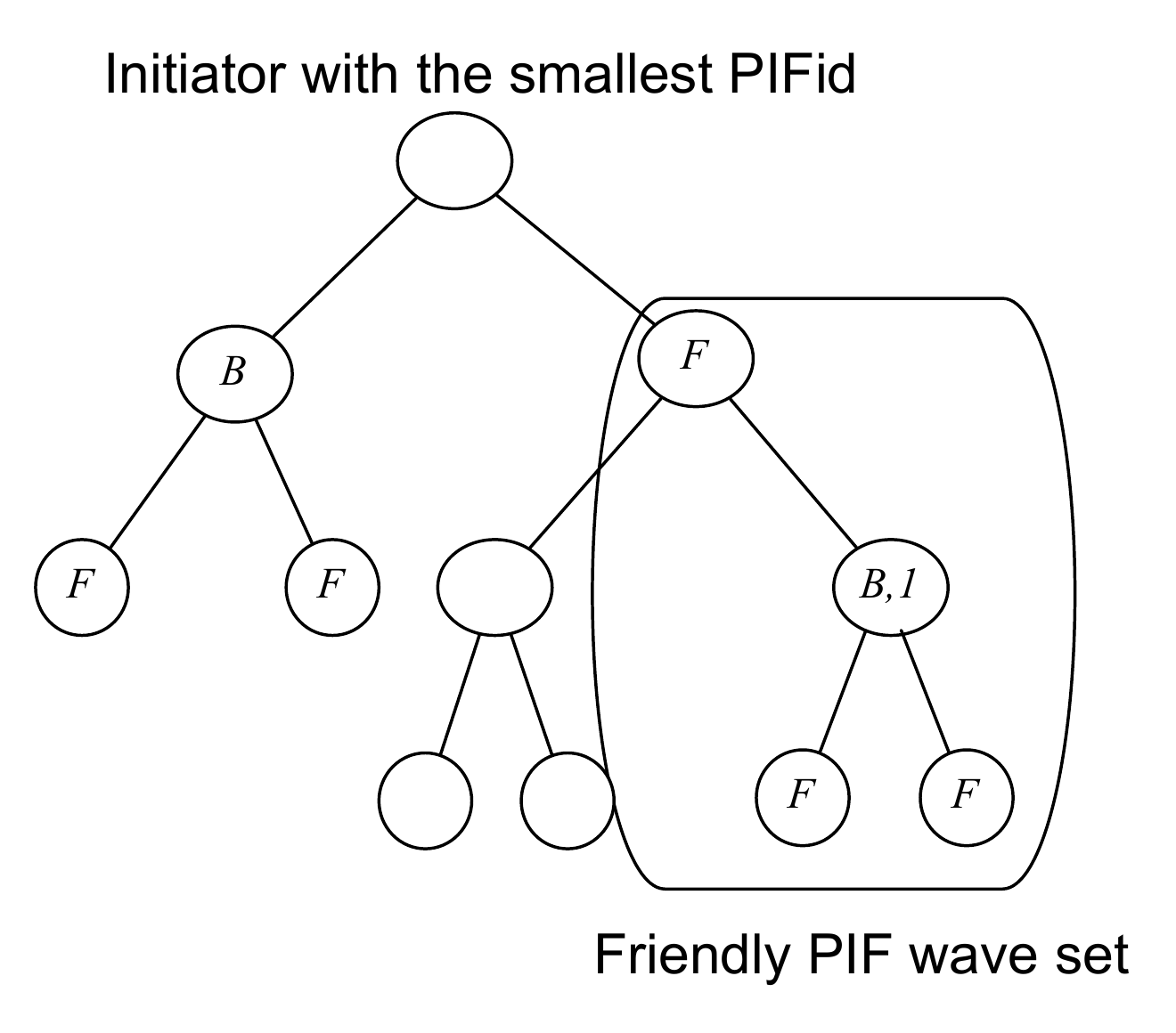}
   \caption{Partial-Final sequence in a $PIF$ Friendly Set}\label{PartialFig}
\end{figure} 

We can state the follow lemma:
 
 \begin{lem}\label{Partial}
In every $PIF$ friendly set, Partial-Final sequence is reached in finite
 \end{lem}
 
\begin{proof}
Let consider a single $PIF$ friendly set $S_i$. It is clear that if in $S_i$ there is only one $PIF$ wave that is
executed, then the lemma holds since the $PIF$ wave in $S_i$ behaves as the $PIF$ wave described in \cite{DBLP:journals/dc/BuiDPV07}. In the following we suppose that there are at least two $PIF$ waves that are executed in each $PIF$ friendly set. 

Observe that the behavior of $PIF$ waves when they are alone (when they do not meet any other $PIF$ wave) is similar
to the well now $PIF$ schema in \cite{DBLP:journals/dc/BuiDPV07}, thus, we will not discuss their progression in this
case. Let us consider the $PIF$ wave that has the smallest identifier $id$ in the set $S_i$ (let refer to its
initiator by $InitTarget$). The cases bellow are then possible:

\begin{enumerate}
\item \label{caseBCB} There exists a node $p$ such as $S_p=(C,NULL,NULL)$ $\wedge$ $\exists$ $q,q' \in N_p$, $S_q=(B,id,z)$ $\wedge$ $S_{q'}=(B,id',z')$. Note that $id<id'$. In this case $p$ sets its state to the broadcast phase and chooses $q$ as its father $S_p=(B,id,q)$. Note that for the node $q'$ we retrieve case \ref{caseBB}.
\item \label{caseBB} There exists a node $p$ such as $S_p=(B,id',q)$ $\wedge$ $\exists$ $q' \in N_p$, $S_{q'}=(B,id,z)$ $\wedge$ $z\ne p$. Note that  $id<id'$. In this case we are sure that $S_{q}=(B,id'',z)$ with $id'' \leq id'$ and $z\ne p$ (recall that there in no more abnormal sequences). Thus, $p$ changes its state to $S_p=(B,id,q')$. Note that $p$ was part of path $P_{id'}$. From Lemma \ref{PathChange}, all the nodes on the path will updates their state as well. Thus the initiator of the $PIF$ wave of identifier $id'$ will be able to know that there is another $PIF$ wave with a smaller $id$ that is being executed (let refer to such a node by $Init_{id'}$). $Init_{id'}$ updates it state to be part of the $PIF$ wave of identifier $id$ (refer to Rule $R5$). Note that only the node that are on the path of both $p$ and $Init_{id'}$ updates their state. The other nodes (that were part of $PIF_{id'}$) don't have to update their state since when the nodes of the path update their state, they are already part of $SubTree(p)$. \\
Observe that since some of the nodes that are part of $SubTree(p)$ do not change their state, Some $PIF$ waves can be seen as $PIF$ waves with smaller $id$. Refer to Figure \ref{SmallerPIF}. (From the figure (case ($c$)) we can observe that $p'$ is not aware of the presence of the $PIF$ wave of $id=0$ since it did not change its state). Let $p_0$ be the node that is neighbor to the initiator of the $PIF$ wave $PIF1$ that has the smallest $id$. Note that $S_{p_{0}}=(B,id,z)$ ($p$ in Figure \ref{SmallerPIF}) and let $p_k$ be the node that is neighbor to the initiator of the other $PIF$ (the one that can be considered as the $PIF$ wave with the smallest $id$, the $PIF$ wave with $id=1$ in Figure  \ref{SmallerPIF}). Observe that $S_{p_{k}}=(B,id',z')$. Let $P= p_1,p_2,p_3, ... p_{k-1}$ be the sequence of nodes between $p_0$ and $p_k$ such as $S_{p_{i}}=(B,id'',p_{i-1})$ and $id''>id$. Note that on $p_{k-1}$, $R4$ is enabled. When the rule is executed, $p_{k-1}$ updates its state to $S_{p_{k-1}}=(B,id',p_k)$. Note that on $p_{k-2}$, $R4$ becomes enabled thus, $p_{k-2}$  will have the same behavior, it updates its state to $S_{p_{k-2}}=(B,id',p_{k-2})$ and so on. Hence, all the nodes on $P$ except $p_1$ will updates their state and set it at $(B,id',p_{i+1})$. Note that $p_1$ is able to detect the presence of the two $PIF$ waves since it has two neighboring nodes that are part of different $PIF$ waves, it is able to detect that $p_2$ think that the $PIF$ with the identifier $id'$ is the smallest one. Thus, $p_1$ updates its state by changing the identifier of the $PIF$ waves to set it at the smallest one. ($S_{p_{1}}=(B,id,z)$). By doing so, $R4$ becomes enabled on $p_{2}$. When the rule is executed, $p_2$ updates its state and so on. 
\end{enumerate}
Let $p_{init}$ be the node that initiates a $PIF$ wave that has the smallest $id$ within the $PIF$ friendly $S_i$. From the two cases above, we can deduce that all the nodes part of the same $S_i$ will be part of the $SubTree(p_{init})$. Note that when the feedback phase finished its execution, all the nodes in $S_i$ except $p_{init}$ will be in the feedback phase. Thus a Partial-Final configuration is reached and the lemma holds.
\end{proof}

\begin{figure}
 \begin{minipage}[b]{.46\linewidth}
  \begin{center}
  \epsfig{figure=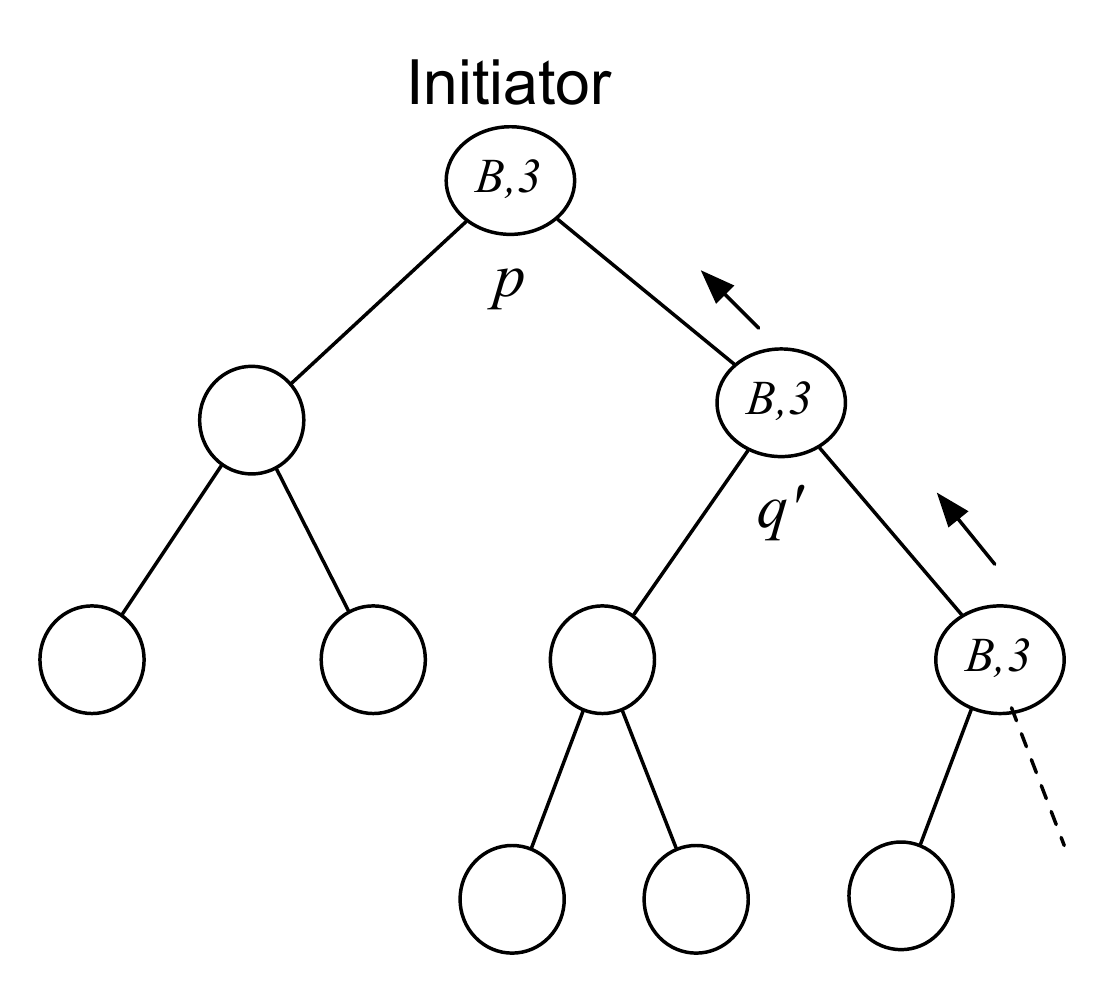,width=4.5cm}\\
  \textit{(a)} A PIF wave with id=3 is being executed.
  \end{center}
 \end{minipage} \hfill
 \begin{minipage}[b]{.46\linewidth}
 \begin{center}
 \epsfig{figure=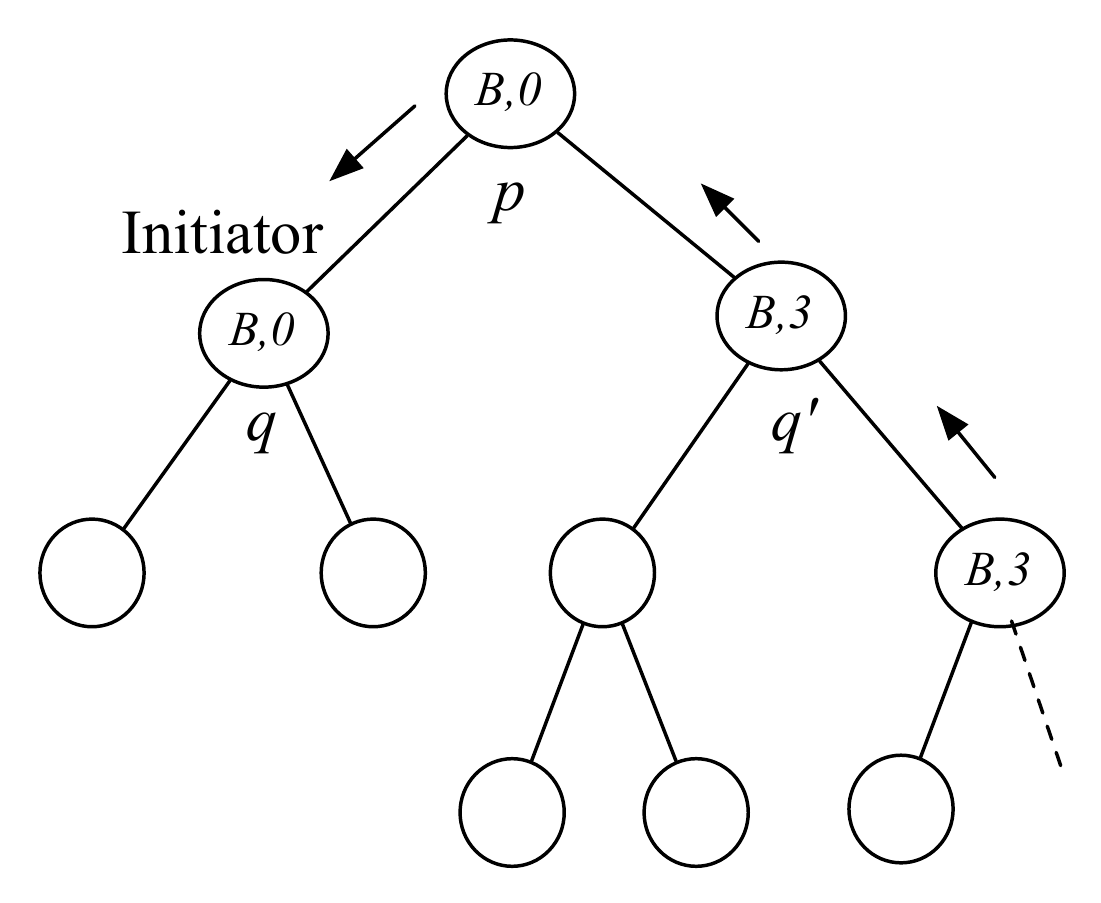,width=4.5cm}\\
 \textit{(b)} a $PIF$ wave with $id=0$ is initiated. p updates its state. 
  \end{center}
 \end{minipage}\hfill
 \begin{minipage}[b]{.46\linewidth}
\begin{center}  
  \epsfig{figure=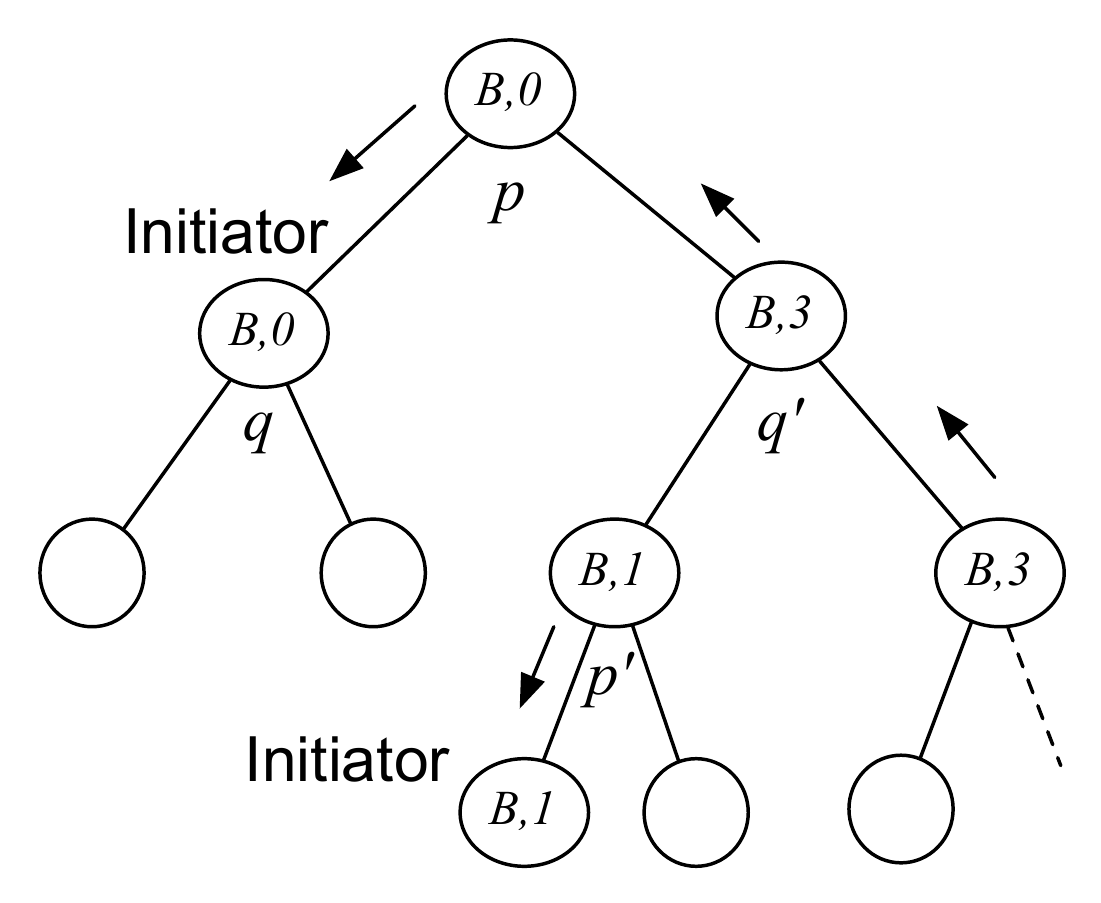,width=4.5cm}\\
 \textit{(c)} A $PIF$ wave with $id=1$ is initiated. $p'$ believe that it is the smallest one and updates its state.
  \end{center}
 \end{minipage}\hfill
  \begin{minipage}[b]{.46\linewidth}
 \begin{center}
 \epsfig{figure=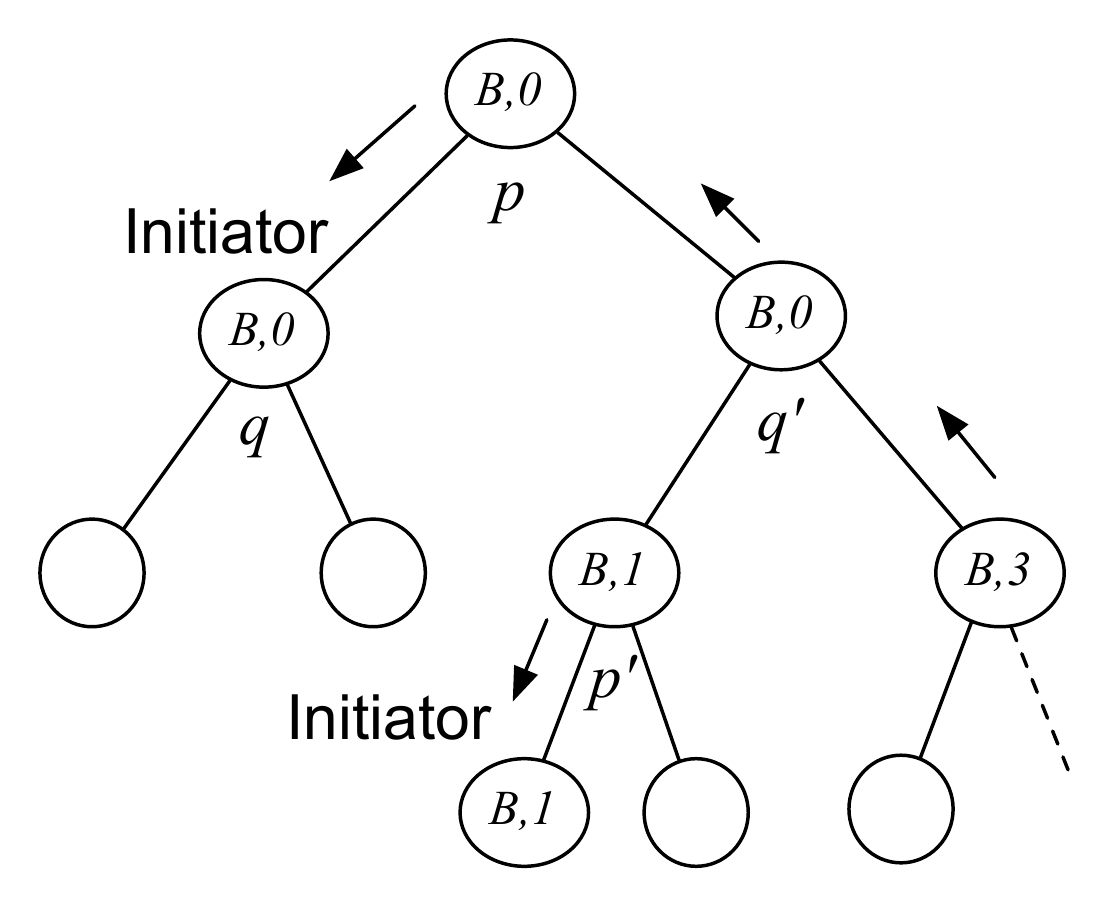,width=4.5cm}\\
  \textit{(d)} q' notices the situation and updates the $id_{PIF}$.
  \end{center}
 \end{minipage}\hfill
 \caption{Special Case.\label{SmallerPIF}}
\end{figure}

From the two cases above, we can deduce that after a finite time, a final-partial configuration is reached in a finite time and the lemma holds.
%Note that since some of the nodes that are part of $SubTree(p)$ do not change their state, Some $PIF$ waves can be seen as $PIF$ waves with smaller $id$. Refer to Figure \ref{SmallerPIF}. According to Lemma \ref{PathChange}, 

 \begin{lem}\label{InitPIF}
Every node belonging to a $PIF$ friendly set eventually clean its state.
 \end{lem}
 
\begin{proof}
Note that when the configuration of type partial-Final is reached in each set, there is only one node $p$ that is enabled. Observe that $S_p=(B,id,NULL)$ (refer to the Rules $R10$ and $R11$). When $p$ executes $R10$, it updates its state to $S_p=(C,NULL,NULL)$. Let $N_1$ be the set of nodes that are neighbor of $p$. Recursively, let $N_i$ be the set of nodes that are neighbor of one node that is part of $N_{i-1}$. Note that $R12$ becomes enabled on all the nodes part of $S1$, when the rule is executed, each node clean its state. In the same manner, $R12$ becomes enabled on the nodes part of $S_2$, and so on. Thus we are sure all the nodes part of each $PIF$ friendly set will eventually clean their state and the lemma holds.

\end{proof}

From Lemma \ref{InitPIF}, all the nodes part of $Initial-PIF$ waves clean their state in a finite time. Thus we are sure that a configuration without any $Initial-PIF$ wave is reached in a finite time.

Let us now consider the system at that time (without $Initial-PIF$ waves). In the following we extend the notion of Partial-Final configuration as follow: we say that a configuration is of type Final-Configuration if there exists in the system a single $PIF$ friendly set such that there exists a node $p$ that verifies the following condition: $S_p=(B,id,NULL)$ and $\forall$ $q \in S_i/\{p\}$, $q \in SubTree(p)$ and $S_q=(F,-,-)$

The following lemma follows:

\begin{lem}\label{FINAL}
If there are many $PIF$ waves that are executed then after a finite time, Final Configuration is reached.
\end{lem}

\begin{proof}
Note that since all the $PIF$ waves that are on the system were initiated by nodes, there is exactly one $PIF$ friendly set \ie there are no $PIF$ waves that are separated by nodes in the feedback phase (since the Feedback phase is initiated from the leaves of the tree overlay and a node $p$ is allowed to change its to the feedback phase only if all its neighboring nodes except his father are already in the feedback phase). Observe that all the nodes of the system are part of the $PIF$ friendly set. Observe also that the partial-Final configuration in this case is exactly the same as the Final configuration. Thus, we can deduce from Lemma  \ref{Partial} that Final configuration is reached in a finite time. and the lemma holds. 
\end{proof}

\begin{lem}\label{FinalClean}
Starting from a final configuration, all the nodes of the system eventually clean their state.
\end{lem}

\begin{proof}
Can be deduce directly from Lemma \ref{InitPIF} (since Partial-Final configuration is the same as Final configuration when there is a single $PIF$ friendly set in the system). Thus, the lemma holds.
\end{proof}

\begin{theorem}\label{generation}
Every node is infinitely often able to initiate a $PIF$ wave
\end{theorem}

\begin{proof}
Directly follows from Lemma~\ref{FinalClean} and the fact that 
%Since each node $p$ of the system cleans its state in finite time (refer to Lemma \ref{FinalClean}). 
if Rule~$R1$ becomes enabled on $p$, then it remains enabled until $R1$ is executed---$R3$ cannot be
enabled while $Request_{PIF}=true$.

%Since each node $p$ of the system cleans its state in a finite time (refer to Lemma \ref{FinalClean}). $R1$ becomes
%enabled on $p$. Note that this rule has the priority on $R3$ since $R3$ is enabled only when $Request_{PIF}=false$.
%When $R1$ is executed on $p$, a $PIF$ wave is initiated and the theorem holds.
\end{proof}

We can now state the following result:

\begin{theorem}
Algorithm $1$ is a self-Stabilizing $CoPIF$ algorithm.
\end{theorem}

\begin{proof}
From Theorem \ref{generation}, each node is able to generate a $PIF$ waves in a finite time. From Lemma \ref{FINAL}, all the nodes of the system were visited by the $CoPIF$ wave. Thus, all of them acknowledge the receipt of the question (whether the tree overlay is in a correct state or not)  and give an answer to the latter. Finally, from Lemma \ref{FINAL} one node $p$ of the system receives the answer ($S_p=(B,id,NULL)$). Hence we can deduce that Algorithm $1$ is a self-Stabilizing $CoPIF$ algorithm and the theorem holds.
\end{proof}

\section{Evaluation}
\label{sec:Eval}

In order to evaluate qualitatively and quantitatively the efficiency of \copif, we drive a set of
experiments. As mentioned before, the \dlpt approach and its different
features have been validated through analysis and simulation \cite{TheseCedric}. 
The scalability and performance of its implementation, \textsc{Sbam} (Spades BAsed Middleware) has ever been improved in \cite{CIT2011}. 
Our goal is now to show the efficiency of the previously described QoS algorithm (Section~\ref{sec:Algorithm}). 
We will focus on the size of the tree, and number of PIF that collaborate simultaneously.
We will observe the behavior not only from the number of exchanged messages point of view but also in term of duration needed to performs \copifs.

\subsection{\textsc{Sbam}}
\label{sec:Sbam}

We use the term \emph{peer} to refer to a physical machine that is available on the network. 
In our case, a peer is an instantiated Java Virtual Machine connected to other
peers through the communication bus. We call \emph{nodes} the vertices of
the prefix tree. 

\textsc{Sbam} is the Java implementation of the \dlpt.  \textsc{Sbam}
proposes $2$-abstraction layers in order to support the distributed
data structure: the \emph{peer}-layer and the \emph{agent}-layer.  The
\emph{peer}-layer is the closest to the hardware layer.  It relies on
the Ibis Portability Layer (IPL) \cite{DrostCaCPE10} that enables the
P2P communication.  We instantiate one JVM per machine, also called
peer. JVM communicate all together as a P2P fashion using the IPL
communication bus.  The \emph{agent}-layer supports the data
structure. Each node of the \dlpt is instantiated as a \textsc{Sbam}
agent. Agents are uniformly distributed over peers and communicate
together in a transparent way using a proxy interface. Since we want
to guarantee truthfulness of information exchanged between
\textsc{Sbam}-agents, the implementation of an efficient mechanism
ensuring quality of large scale service discovery is quite
challenging.  In the state model described in the
section~\ref{ref:CompMod} a node has to read the state and the
variables of its neighbors.  In \textsc{Sbam}, the feature is
implemented using synchronous message exchange between agents.
Indeed, when a node has to read its neighbor states, it sends a
message to each and wait all responses.  Despite the fact that this
kind of implementation is expensive, especially on a large distributed
data structure, experiment (Section~\ref{sec:results}) shown that our
model implementation stays efficient, even on a huge prefix tree.

\begin{figure*}
  \centering
  \subfigure[\label{fig:evolution_nbMess} Evolution of number of exchanged messages.]{\mbox{\includegraphics[trim = 0mm 5mm 10mm 16mm, clip, width=.35\linewidth]{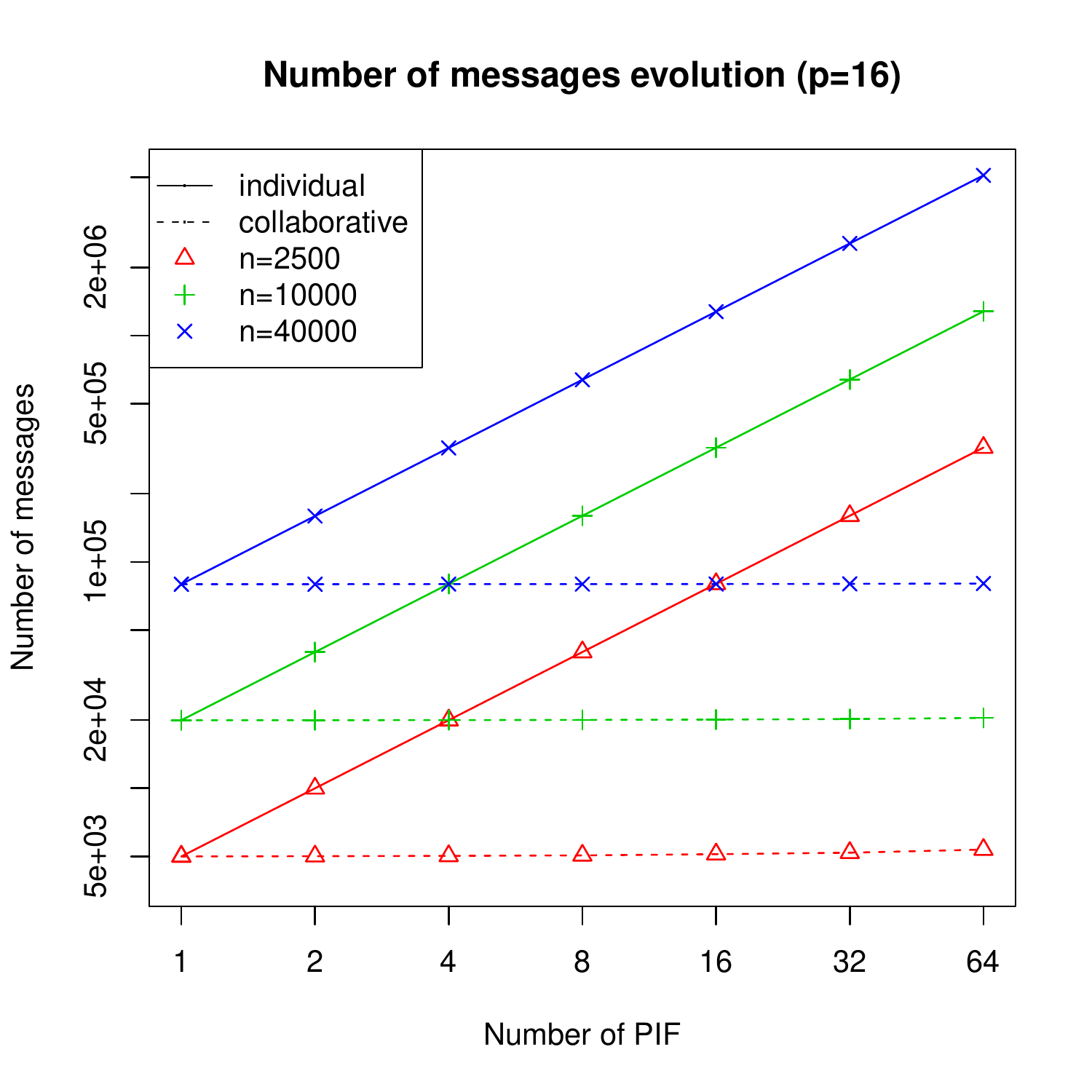}}}
  \subfigure[\label{fig:evolution_medDur} Evolution of duration.]{\mbox{\includegraphics[trim = 0mm 5mm 10mm 16mm, clip, width=.35\linewidth]{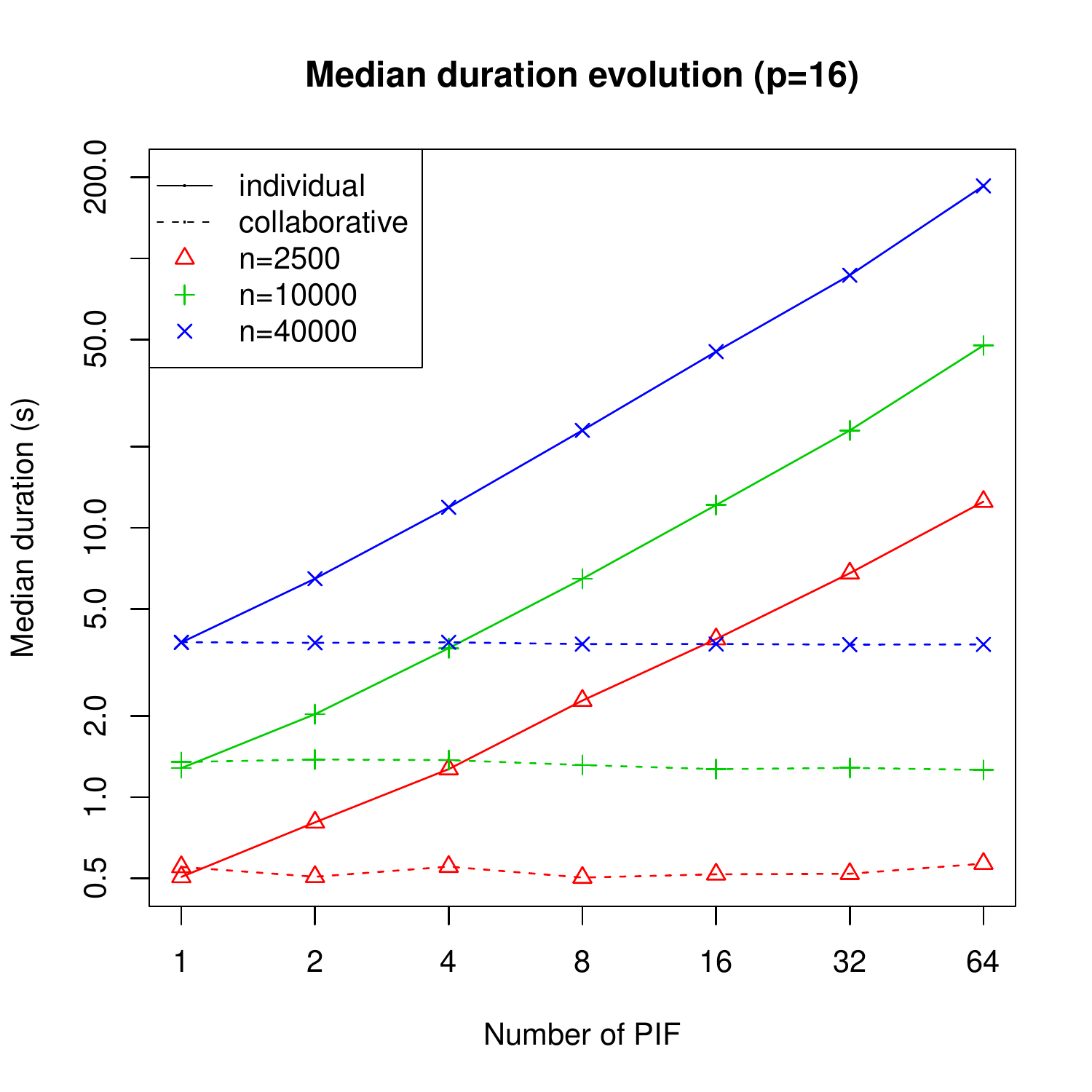}}}
  \caption{\label{fig:behaviour} \copif behavior}
\end{figure*}

\subsection{Experimental platform}
\label{sec:exp_pf}

Experiments were run on the Grid'5000 platform\footnote{\url{http://www.grid5000.fr/}} \cite{BCC+06},
 more precisely on a dedicated cluster 
\emph{HP Proliant DL165 G7} $17$ units, each of them equipped with 2
\emph{AMD Opteron 6164 HE} (1.7GHz) processors, each processor gathering 12
cores, thus offering a 264-cores platform for these experiments. Each unit
consists of 48 GB of memory. Units are connected through two~Gigabit Ethernet
cards. For each experiment, we deployed one peer per
unit.

\subsection{Scenario of experiments}

The initialization of an experiment works in three phases: ($i$) the communication bus is started on a computing unit
(Section~\ref{sec:exp_pf}), ($ii$) $16$ peers are launched and connected together through the communication bus, and
($iii$) a pilot is elected using the elect feature of the communication bus.

After the initialization, the pilot drives the experiment. It 
consists in two sequential steps. First, it sends  $n$ insertion requests to
the distributed tree structure.
An insertion request leads to the addition of a new entry in the \dlpt
tree (Section~\ref{sec:P2PServiceDiscoveryFramework}). The insertion requests are sent to a random node of the tree and routed
following the lexicographic pattern to the targeted node. Doing so, the node
sharing the greatest common prefix with the service name is reached. If the
targeted label does not exist, a new node is created on a
randomly chosen peer and linked to the existing tree. 

In the final step, the pilot selects a set of nodes to initiate {\it classic}-PIFs and \copifs (Section~\ref{sec:Algorithm}). 
In order to observe distributions, $10$ replications of this basic scenario are executed.

\subsection{Failures}
\label{sec:failures}

At this level of description we can distinguish two kind of failure: ($i$) failures in the \dlpt data structure, when
the prefix tree data structure is corrupted; ($ii$) failures in the \copif state variables, when state variable dedicated to \copif feature are corrupted. 

In our experiments, we consider that the \dlpt data structure is not corrupt. 
It corresponds to the worst case in term of \copif truthfulness check. 
Indeed, if the entire \dlpt data structure is correct, the \copif has to explore the entire \dlpt data structure to check it.

Next, remind that the objective of our experiment is to evaluate the efficiency of \copif compared to \textit{classic}-PIF. 
So, we consider that the \copif state variables are not corrupted and we measure the number of messages and the duration of \copif when the self-stabilization \copif has converged, it means after the clean part of the \copif state variables.

\subsection{Parameters and indicators}

\label{sec:ParamAndIndic}

The experiments conducted are influenced by two main parameters. 
First, $n$ denotes the number of inserted services in the tree. 
Second, $k$ refers to the number of PIF waves that are collaborating together. 

In these experiments, three trees were created with $n$ in the set $\{2500, 10000,$\\
 $40000\}$. The number of PIF that
collaborate ($k$) was taken from the set $\{1, 2, 4, 8,$\\
$ 16, 32, 64\}$. For each couple $(n,pif)$, $10$ replications are performed. Thus, $210$ experiments were conducted.

Strings used to label the nodes of the trees were
randomly generated with an alphabet of $2$ digits and a maximum length of $18$
%{\it i.e.} 
(in a set of $524287$ key). 

For each experiment we observe two indicators: ($i$) the {\it total number of exchanged messages} observed and ($ii$) the
{\it time required} to perform $k$ {\it classic}-PIFs or \copifs over distributed data structure, {\it i.e.}, the time between the issue of the
PIFs and the receipt of the response on all nodes that initiate PIFs. 

For each indicator we obtain $21$ sets of $10$ values. 
We present evolution of the $median$-value of the $10$-replication according to $k$, the number of PIFs
(Figures~\ref{fig:evolution_nbMess} and~\ref{fig:evolution_medDur}).
The comparison of these indicators for {\it classic}-PIFs and \copifs provides us a \textbf{qualitative} overview of the gain obtain using \copif . %strategy.

In order to \textbf{quantitatively} evaluate the efficiency of the \copif strategy, for an indicator ($I$) and for a
given number of PIFs $k$, we compute the \textbf{efficiency criterion} with the following formula:
$$ E_{I,pif} = \frac{I_{\text{ind},pif}}{pif \times I_{\text{coll},pif}},$$ 
where $I_{\text{classic},pif}$ (resp. $I_{\text{CoPIF},pif}$) is the value of the indicator $I$ for a given $k$ and in an {\it classic}-PIF (resp. \copif) context. The evolution of this efficiency criterion are shown in Figures~\ref{fig:efficiency_nbMess}~and~\ref{fig:efficiency_medDur}.

\subsection{Results}
\label{sec:results}

Figure~\ref{fig:evolution_nbMess} (resp. \ref{fig:evolution_medDur}) presents the evolution of the number of messages
(resp. duration) needs to execute PIFs according to number of PIFs ($k$) that are simultaneously performed and the size of the data structure on which PIFs are performed. The $y$-axis represents the number of exchanged messages (resp. the duration).
On these figures, {\it classic}-PIFs and \copifs strategies are compared. 
On both curve, the $x$-axis represents the number of PIF ($k$) that are simultaneously performed.
On these figures, we present $2$ curve triplets. Solid (resp. dashed) curves triplet describes indicator in an {\it classic}-PIF (resp. \copif) context. 
For each triplet, \emph{red-triangle}-curve (resp. \emph{green-$+$}-curve and \emph{blue-$x$}-curve) describes behavior of indicator for $n=2500$ (resp. $n=10000$ and $n=40000$).

When the indicators explode in for {\it classic}-PIF strategy, they stay stable for \copif strategy. It \textbf{qualitatively} demonstrates gain of the \copif strategy over {\it classic}-PIF approach.
The introduction of the {\it efficiency criterion} in Section~\ref{sec:ParamAndIndic} allows us to measure quantitatively this gain and its behavior.

Figure~\ref{fig:efficiency_nbMess}~and~\ref{fig:efficiency_medDur} present the efficiency of \copif according to the number of exchanged messages and the duration. On these figures we want to observe the impact of the size of the data structure on the efficiency of \copif.
Figure~\ref{fig:efficiency_nbMess} reveals us that, in term of number exchanged messages, on small data structure, the collaborative mechanism is less efficient than on huge one. 
This result was expected because on small data structure the number of messages due to collision between collaborative PIFs (overhead) represents a more important part of the entire number of exchanged messages. So, the bigger the data structure, the more efficient \copif.

More interesting is the analyze of the Figure~\ref{fig:efficiency_medDur}. Indeed we can observe the same tendency in term of duration but the efficiency decrease faster with the number of PIFs that are simultaneously performed. It is explain by the fact that overhead messages introduced by \copif are particularly expensive  messages in term of duration. It \textbf{quantitatively} demonstrates gain of the \copif strategy over {\it classic}-PIF approach.

\begin{figure*}
  \centering
  \subfigure[\label{fig:efficiency_nbMess} Number of PIF according to the number of exchanged messages.]{\mbox{\includegraphics[trim = 0mm 5mm 10mm 16mm, clip, width=.35\linewidth]{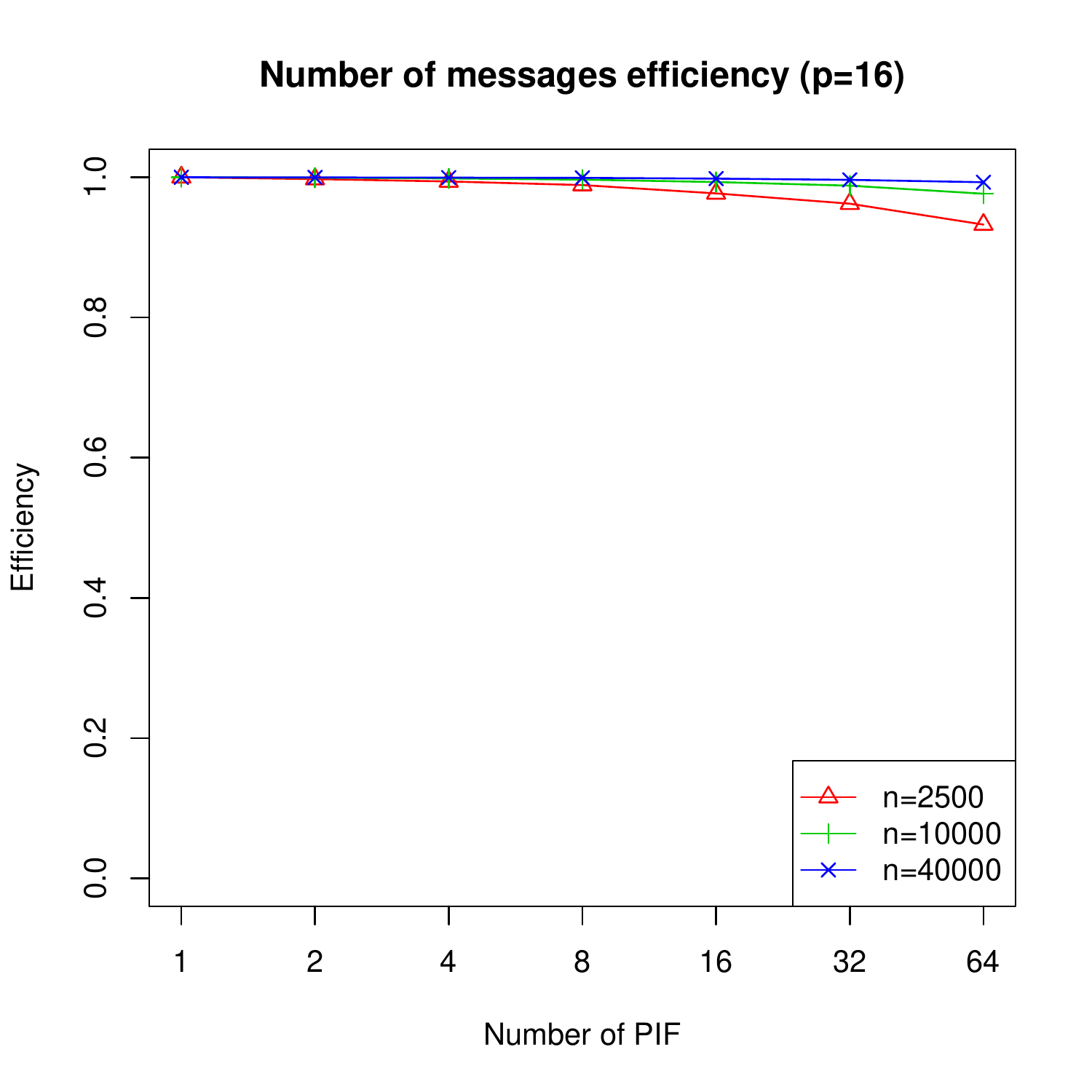}}}
  \subfigure[\label{fig:efficiency_medDur} Number of PIF according to the duration.]{\mbox{\includegraphics[trim = 0mm 5mm 10mm 16mm, clip, width=.35\linewidth]{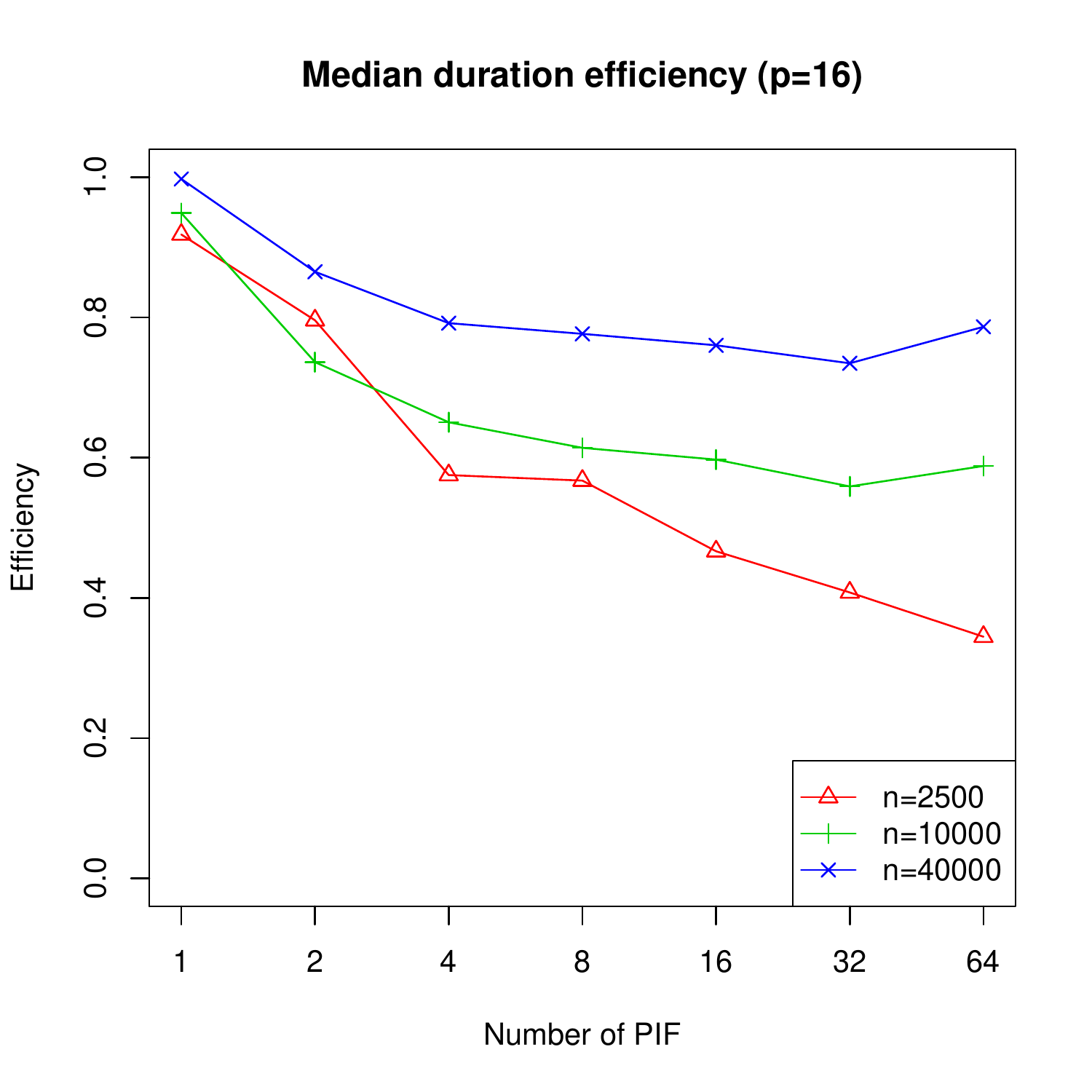}}}
  \caption{\label{fig:efficiency} \copif Efficiency }
\end{figure*}

\section{Conclusion and Future Work}
\label{sec:Conclusion}

In this paper we provide a self-stabilized collaborative algorithm called {\em \copif} allowing to 
check the truthfulness of a distributed prefix tree.  
%A complete correctness proof is given. %Convergence is ensure in XX\todo{F/A: conv. en ??}.
\copif implementation in a P2P service discovery framework is experimentally validate, in \emph{qualitative} and in \emph{quantitative} terms.
Experiment demonstrates the efficiency of \copif w.r.t. {\it classic}-PIF.
\copif overhead represents a small part of the number of exchange messages and of the time spend, 
specially on huge data structures.

We conjecture that the stabilization time is in $O(h^2)$ rounds and the worst case time to merge several classic-PIF
waves is in $O(h)$ rounds, $h$ being the height of the tree. 
We plan to experimentally validate this two complexities. Indeed, experiment were driven considering 
no corrupted \copif variables. %Assuming that clean step converge in XX\todo{F/A: conv. en ??}.  
In order to do that, we need to define a model of failure, implement or reuse a fault injector 
and couple it with \textsc{Sbam} before driving a new experiment campaign.

\section*{Acknowledgment}

This research is funded by french National Research Agency
(08-ANR-SEGI-025). Details of the project on
\url{http://graal.ens-lyon.fr/SPADES}. Experiments presented in this
paper were carried out using the Grid'5000 experimental testbed, being
developed under the INRIA ALADDIN development action with support from
CNRS, RENATER and several Universities as well as other funding bodies
(see \url{https://www.grid5000.fr}).

\bibliographystyle{plain}
\bibliography{COPIF}

%%%%%%%%%%%%%%%
\clearpage

\appendix

\end{document}